\let\latexdocument\document
\let\latexarabic\arabic
 
\documentclass{article}
\usepackage[margin=1.5in]{geometry}

\let\document\latexdocument
\let\arabic\latexarabic

\usepackage{times}
\usepackage{bm}
\usepackage{natbib}

\usepackage{subcaption}
\usepackage{adjustbox}

\usepackage{pifont}
\makeatletter
\newcommand\Pimathsymbol[3][\mathord]{%
  #1{\@Pimathsymbol{#2}{#3}}}
\def\@Pimathsymbol#1#2{\mathchoice
  {\@Pim@thsymbol{#1}{#2}\tf@size}
  {\@Pim@thsymbol{#1}{#2}\tf@size}
  {\@Pim@thsymbol{#1}{#2}\sf@size}
  {\@Pim@thsymbol{#1}{#2}\ssf@size}}
\def\@Pim@thsymbol#1#2#3{%
  \mbox{\fontsize{#3}{#3}\Pisymbol{#1}{#2}}}
\makeatother
\DeclareFontFamily{U}{psyr}{}
\DeclareFontShape{U}{psyr}{m}{n}{<-> psyr }{}

\DeclareFontFamily{U}{psyro}{}
\DeclareFontShape{U}{psyro}{m}{n}{<-> psyro }{}

\newtheorem{assumption}{Assumption}
\newtheorem{remark}{Remark}
\newtheorem{definition}{Definition}
\newtheorem{theorem}{Theorem}
\newtheorem{lemma}{Lemma}
\newenvironment{proof}{\paragraph{Proof:}}{\hfill$\square$}

\usepackage{hyphenat}
\usepackage{bbold}
\usepackage{bbm}
\usepackage[T1]{fontenc}
\usepackage{amsmath}
\usepackage{amssymb}
\usepackage{prodint}
\usepackage{color}

\usepackage{mathtools}
\usepackage[usestackEOL]{stackengine}

\usepackage{accents}

\newcommand{\cadlag}{c\`{a}dl\`{a}g }

\usepackage{amsmath}
\usepackage{mathrsfs}
\newcommand{\EE}{\mathbb{E}}

\newcommand{\eps}{\varepsilon}

\newcommand{\R}{\mathbb{R}}

\newcommand{\1}{\mathbb{1}}

\newcommand{\tmax}{\tau}

\usepackage[plain,noend]{algorithm2e}

\makeatletter
\renewcommand{\algocf@captiontext}[2]{#1\algocf@typo. \AlCapFnt{}#2} 
\def\@algocf@capt@plain{top}
\renewcommand{\algocf@makecaption}[2]{%
  \addtolength{\hsize}{\algomargin}%
  \sbox\@tempboxa{\algocf@captiontext{#1}{#2}}%
  \ifdim\wd\@tempboxa >\hsize
    \hskip .5\algomargin%
    \parbox[t]{\hsize}{\algocf@captiontext{#1}{#2}}
  \else%
    \global\@minipagefalse%
    \hbox to\hsize{\box\@tempboxa}
  \fi%
  \addtolength{\hsize}{-\algomargin}%
}
\makeatother

\begin{document}

\title{One-step TMLE for targeting cause-specific absolute risks and
  survival curves\\[0.5cm]
\textit{Tech report}}

\author{\small Helene C. W. Rytgaard$^{1,*}$ and Mark J. van der Laan$^{2}$ \\
  \small $^{1}$Section of Biostatistics, University of Copenhagen,
  Denmark\\ \small$^{2}$Devision of Biostatistics, University of
  California, Berkeley}

\maketitle

\begin{abstract}
  This paper considers one-step targeted maximum likelihood estimation
  methodology for general competing risks and survival analysis settings
  where event times take place on the positive real line \(\R_+\) and
  are subject to right-censoring.  Our interest is overall in the
  effects of baseline treatment decisions, static, dynamic or
  stochastic, possibly confounded by pre-treatment covariates. We
  point out two overall contributions of our work. First, our method
  can be used to obtain simultaneous inference across all absolute
  risks in competing risks settings. Second, we present a practical
  result for achieving inference for the full survival curve, or a
  full absolute risk curve, across time by targeting over a fine
  enough grid of points. The one-step procedure is based on a
  one-dimensional universal least favorable submodel for each
  cause-specific hazard that can be implemented in recursive steps
  along a corresponding universal least favorable submodel. We present
  a theorem for conditions to achieve weak convergence of the
  estimator for an infinite-dimensional target parameter. Our
  empirical study demonstrates the use of the methods.
\end{abstract}

\section{Introduction}

This work proceeds on the basis of the work of \citet[][Chapter
5]{van2016one,van2018targeted} to construct a one-step targeted
maximum likelihood estimation procedure for survival and competing
risks settings with event times taking place on the positive real line
and interest is in the effect of a treatment assigned at baseline
adjusted for baseline covariates. This is the one-step version of the
iterative targeted maximum likelihood estimation method presented in
\cite{rytgaard2021estimation}. For the present paper, we point out two
important contributions to causal inference in survival and competing
risks analysis:
\begin{enumerate}
\item We present a method for analyzing treatments effects on all
  state occupation probabilities in competing risks settings
  simultaneously. Specifically, our method can be used to achieve
  multivariate inference across all of the cause-specific absolute
  risks.
\item Our method provides a practical procedure to obtain weak
  convergence and asymptotic efficient estimation of the full survival
  curve, or a cause-specific absolute risk curve, across all
  time-points in an interval of \(\R_+\).
\end{enumerate}

We emphasize our contribution to competing risks analysis
settings. Here, interest is often in the absolute risk, or the
subdistribution, of a single event type of interest; however,
providing inference solely for a single subdistribution may lead to
obscure conclusions like a medical treatment in fact killing patients
being a protective treatment for the event of interest. This problem
is not solved by considering cause-specific hazard functions instead,
as these generally fail to have causal interpretation
\citep{hernan2010hazards,martinussen2018subtleties}. Alternative
one-dimensional causal estimands in competing risks analysis are
discussed in recent work by \cite{young2018causal} and
\cite{stensrud2019separable} who propose and distinguish different
approaches to isolate effects on a single absolute risks of interest
relying on further untestable assumptions on the nature of the
treatment mechanism. In our work, we propose an alternative route to
carry out a complete analysis of all absolute risks simultaneously;
this will reveal all patterns and give the full picture of the effects
of a given treatment.

Targeted maximum likelihood estimation
\citep{van2006targeted,van2011targeted} is a general methodology for
semiparametric efficient substitution estimation
\citep{bickel1993efficient,vanRobins2003unified,tsiatis2007semiparametric}
of causal parameters consisting of two steps. First, flexible machine
learning methods are applied to estimate high-dimensional nuisance
parameters. This is followed by a targeting step applied to update the
initial estimators to solve a desired equation of interest, the
efficient influence curve equation. This second step yields asymptotic
linearity, double robustness and efficiency under weak conditions on
the statistical model assumed for the data-generating distribution
\citep{van2017generally}. Our focus in this work is on the targeting
step; methods for flexible initial estimation is discussed in previous
work, see, e.g., \cite{rytgaard2021continuous,rytgaard2021estimation}.

The usual targeting procedures are constructed based on local least
favorable submodels \citep{van2000asymptotic,van2011targeted} along
which current estimators are updated iteratively to finally solve the
efficient score equation of interest. A one-step targeting procedure,
on the other hand, is constructed based on universal least favorable
submodels, as introduced by \cite{van2016one}, based on which the
updating scheme only requires a single iteration. In addition, the
universal least favorable submodels provide a method to solve multiple
score equation simultaneously, again with only a single update
step. As shown by \cite{van2016one}, the maximum likelihood estimator
based universal least favorable submodels can be found by recursive
infinitesimal updating steps along a local least favorable submodel.

In this work, we construct a one-step targeting procedure to target
both multivariate and infinite-dimensional target parameters in
survival and competing risks settings. Particularly, we present a
methodology for targeting all cause-specific subdistribution
simultaneously and for targeting both subdistributions and survival
curves across multiple time-points. We further show that we can get
inference for the full survival curve across time by targeting over a
grid which is fine enough.  Due to the simultaneous targeting, the
one-step targeted maximum likelihood estimator is guaranteed to yield
monotone survival curves and cause-specific absolute risks that are
guaranteed to sum up to one, thus completely respecting the parameter
space constraints of the problem. Furthermore, the simultaneous
inference we provide for the target parameter yields a direct
methodology for multiple testing correction. We remark that focus in
the current field of semiparametric efficient estimation for survival
and competing risks analysis has been on efficient estimation of
one-dimensional parameters, corresponding to solving just a single
score equation
\citep{hubbard2000nonparametric,moore2009covariate,stitelman2011targeted2,benkeser2018improved,ozenne2020estimation,rytgaard2021estimation}.
To use these methods to construct efficient estimators for the
multivariate target parameters would require separate estimation of
each one-dimensional component of the target parameter.

In our one-step procedure, one has the choice between different 
Hilbert space norms which guide the direction of the one-step targeted
update. We consider and discuss three different choices and explore in
a simulation study their finite-sample performance compared to each
other and to the iterative targeted maximum likelihood estimation
method from \cite{rytgaard2021estimation}.

\section{Setting and notation}
\label{sec:intro:setting:notation}

We consider a competing risks setting with \(J \ge 1\) causes with observed
data on the form
\begin{align*}
  O= (L, A, \tilde{T}, \tilde{\Delta}) \in \R^{d'} \times  \lbrace 0,1\rbrace \times
  \R_+ \times \lbrace 0, 1, \ldots, J\rbrace, 
\end{align*}
where \(A \in \lbrace 0,1\rbrace\) is a baseline treatment variable,
\(L \in \R^{d'}\) are pretreatment covariates, \(\tilde{T}\in \R_+\)
is the observed time under observation and
\(\tilde{\Delta} \in \lbrace 0, 1, \ldots, J\rbrace\) is an indicator
of right-censoring (\(\tilde{\Delta} = 0\)) or type of event
(\(\tilde{\Delta} \ge 1\)) observed. The variables \(T\in \R_+\) and
\(C\in \R_+\) represent the times to event, one of \(J\ge 1\) types,
and censoring, respectively, such that the observed time is
\(\tilde{T} = \min ( T, C)\) and
\(\tilde{\Delta} = \1\lbrace T \ge C \rbrace \Delta\) where \(\Delta\)
is the uncensored event indicator. The special-case with \(J=1\)
corresponds to a classical survival analysis setting. Let \(P_0\)
denote the distribution of \(O\) and assume that \(P_0\) belongs to
the nonparametric statistical model \(\mathcal{M}\). For
\(j=1,\ldots, J\), let \(\lambda_{0,j}\) denote the cause \(j\)
specific hazard, defined as
\begin{align*}
  \lambda_{0,j}(t \, \vert \, a, \ell)
  & =
    \underset{h \rightarrow 0}{\lim} \,\, h^{-1} P(T \le t+h,  \Delta=j \mid {T} \ge t,
    A=a, L=\ell),  
\end{align*}
and let \(\Lambda_{0,j}(t \mid a,\ell)\) denote the corresponding
cumulative hazard. Likewise, let \(\lambda_0^c(t \mid a,\ell)\) denote
the conditional hazard for censoring and
\(\Lambda_0^c(t \mid a,\ell)\) the corresponding cumulative hazard.

Let further \(\mu_0\) be the density of \(L\) with respect to an
appropriate dominating measure \(\nu\) and
\(\pi_0(\cdot \, \vert \, L) \) be the conditional distribution of
\(A\) given \(L\). The distribution for the observed data can now be
represented as \( dP_0(o) = p_0(o) d\nu(\ell) dt\) where
\(o = (\ell, a, t, \delta)\) and the density \(p_0\) under coarsening
at random \citep{vanRobins2003unified} factorizes as follows
\begin{align}
  \begin{split}
    p_0(o) = \mu_0(\ell) \pi_0(a\, \vert \, \ell) \big(
    {\lambda}^c_0 (t \mid a, \ell) \big)^{\1\lbrace\delta = 0\rbrace} 
    {S}_0^c(t - \mid a, \ell)   \prod_{j=1}^J  \big(
    {\lambda}_{0,j} (t \mid a, \ell) \big)^{\1 \lbrace \delta = j\rbrace}
    {S}_0(t- \mid a, \ell);
\end{split}
  \label{eq:survival:factorization}
\end{align}
here, still under coarsening at random,
\begin{align*}
  {S}_0 (t \mid a, \ell)
  & = \exp \bigg(- \int_0^t \sum_{j=1}^J
    {\lambda}_{0,j}(s \mid a, \ell) ds\bigg) ,
    \intertext{ is the survival function and }
    {S}^c_0 (t \mid a, \ell)
  &= \exp \bigg(- \int_0^t 
    {\lambda}_0^c(s \mid a, \ell) ds\bigg), 
\end{align*}
is the censoring survival function. For \(j=1,\ldots, J\) we further denote by 
\begin{align*}
  {F}_{0,j}(t \mid a,\ell)
  = \int_0^t {S}_{0}(s- \mid a, \ell) {\lambda}_{0,j} (s \mid a, \ell) ds, 
\end{align*}
the absolute risk function, or the subdistribution, for events of type
\(j\) \citep{gray1988class}. We consider the nonparametric statistical
model \(\mathcal{M}\) such that the density \(p\) of any
\(P\in \mathcal{M}\) factorizes as in
\eqref{eq:survival:factorization}.

\section{Multivariate target parameters}
\label{sec:intro:parameter}

We start by considering the case where interest is in estimation of a
multivariate target parameter taking values in a Hilbert space
\(\mathscr{H}_{\Sigma_d}\) of elements \(x \in \R^{d} \) endowed with
inner product and corresponding norm
\begin{align}
  \langle x, y \rangle_{\Sigma_d} = x^\top \Sigma_d^{-1} y,
  \qquad  \Vert x \Vert_{\Sigma_d} = \sqrt{x^\top \Sigma_d^{-1} x} ,
  \label{eq:define:inner:product}
\end{align}
for a user-supplied  positive definite matrix
\(\Sigma_d \in \R^{d} \times \R^{d}\).  Particularly, let now
\(\Psi \, : \, \mathcal{M} \rightarrow \mathscr{H}_{\Sigma_d}
\subset\R^d\), where \(d=JK\), be the multivariate target parameter
with components given by
\begin{align}
  \Psi_{j,k} (P) = \EE \big[ F_j ( t_k \mid A=a^*, L)\big] 
  = \EE \big[ P( T \le t_k , \Delta = j \mid A=a^*, L)\big], 
  \label{eq:Psi:j:k}
\end{align}
for \( j=1, \ldots, J\) and \( k=1, \ldots, K\).  Here \(a^*\) could
be either 1 or 0, to target the treatment or control specific
probabilities. Note that taking the difference between the two
corresponds to the average treatment effect. Causal assumptions of
consistency, positivity and no unmeasured confounding yields a causal
interpretation of \eqref{eq:Psi:j:k} as the absolute risk of events of
type \(j\) before time \(t_k\) had everyone in the population,
possibly contrary to fact, been assigned to treatment level \(A=a^*\)
\citep{hernanrobins,rytgaard2021estimation}. As is well-known, the
target parameter in \eqref{eq:Psi:j:k} can also be written as
\begin{align*}
  \Psi_{j,k} (P) = \int_{\mathcal{L}} F_1(\tmax \mid
  a^*, \ell)  \mu ( \ell)d\nu(\ell)
  = \int_{\mathcal{L}} \int_0^{\tau} S(s- \mid a^*, \ell) \Lambda_1 (ds \mid a^*, \ell)   \mu ( \ell)d\nu(\ell),
\end{align*}
seen to depend on all cause-specific hazards via
\(S(t \mid a^*,\ell) = \exp (- \int_0^t \sum_{j=1}^J {\lambda}_{j}(s
\mid a^*, \ell) ds)\).
    
\begin{remark}[Euclidean norm] 
  The user-supplied positive definite matrix
  \(\Sigma_d \in \R^{d} \times \R^{d}\) could be chosen as the
  identity matrix so that the norm in \eqref{eq:define:inner:product}
  is simply the standard Euclidean norm. This is the Hilbert space
  considered and implemented in previously proposed one-step
  procedures \citep{van2016one,cai2018one}. In this work, we will
  consider other alternative choices of norms as well, see Section
  \ref{sec:implementation}.
  \label{rem:choice:Sigma:euclidean:norm}
\end{remark}

\subsection{Efficient influence function}
\label{ssec:eff:ic}

The efficient influence function for the \((j,k)\)-specific component
of the target parameter is given by, see, e.g.,
\cite{rytgaard2021estimation},
\begin{align*}
  D_{j, t_k}^*(P) (O) = \sum_{l=1}^J \int   h_{j,l,k,t} (P) (O) \, \big( N_l(dt) -
  \1 \lbrace \tilde{T}\ge t\rbrace \lambda_l(t \, \vert \, A,L) dt \big)  \\
  + \,
  F_{j}(t_k \, \vert \, a, L)
  -   \Psi_{j,k}(P), 
\end{align*}
with the functions \( h_{j,l,k,t} \) defined by
\begin{align*}
  h_{j,l,k,t} (P)(O) =
  \frac{
  \1 \lbrace A=a\rbrace }{ \pi (A \, \vert \, L) } \frac{\1 \lbrace t  \le t_k \rbrace}{
  S^c( t- \, \vert \, A,L)
  } \begin{cases}
    1- \frac{ F_{l} ( t_k  \mid A, L) - F_{l}( t \mid A, L)}{
      S(t \, \vert \, A,L)}, &\text{when } \, l=j, \\
    -\frac{ F_j ( t_k  \mid A, L) - F_j ( t \mid A, L)}{
      S(t \, \vert \, A,L) },  &\text{when } \, l \neq j,
  \end{cases}
\end{align*}
characterizing the least favorable paths for the estimation problem.
We introduce a vectorized notation and use
\begin{align*}
  D^* &= (D^*_{j,t_k} \, : \, j=1, \ldots, J, k=1, \ldots, K), \\
  h_{l, t} &= (h_{j,l,k,t} \, : \, j=1, \ldots, J, k=1, \ldots,
                K),
\end{align*}
to refer to the \(d\)-dimensional vector of stacked efficient
influence functions and functions indexing the least favorable paths,
respectively.

\subsection{Nuisance parameters for the estimation problem}
\label{sec:nuisance}

We note that the target parameter depends only on the cause-specific
hazards \(\lambda = (\lambda_l \, : \, l = 1, \ldots, J)\) and the
covariate density \(\mu\) whereas the efficient influence function is
a mapping \(P\) through \(\lambda\) and \(\mu\) as well as the
treatment distribution \(\pi\) and the censoring survival function
\(S^c\). Construction of an efficient estimator requires estimation of
all these quantities. To reflect this in our notation we will use the
alternative notation
\begin{align*}
  \tilde{\Psi} (\lambda) := {\Psi} (P) ,
  \quad
  \tilde{D}^* (\lambda, \pi, S^c) (O) :=     {D}^* (P) (O), \quad
  \text{and,} \quad \tilde{h}_{l, t} (\lambda,\pi,S^c) :=h_{l,t} (P),
\end{align*}
when appropriate; note that we have suppressed the dependence on
\(\mu\) in this notation, as we will simply estimate the average over
the covariate distribution by the empirical average. We make the
following assumptions on the nuisance parameters.

\begin{assumption}[Conditions on \(\mathcal{M}\)]
  Assume that the nuisance parameters
  \(\lambda_1,\ldots, \lambda_J,\lambda^c, \pi\) can be parametrized
  by functions that are \cadlag (continuous from the right, limits
  from the left) and have finite sectional variation norm
  \citep{gill1995inefficient,van2017generally}. Assume further
  positivity, i.e., \(S^c(\tau \mid a, L) \pi(a \mid L) > \kappa >0\),
  for \(a=0,1\), and lastly that \(S(\tau\mid A, L) > \kappa' >0\).
  \label{ass:M:cadlag}
\end{assumption}

Assumption \ref{ass:M:cadlag} particularly allows for the construction
of highly adaptive lasso estimators \citep
{benkeser2016highly,van2017generally} for each of the nuisance
parameters.  For the present paper we simply assume that we have at
hand a set of initial estimators
\(\hat{\lambda}_n, \hat{\pi}_n, \hat{S}_n^c\); for details, see
\cite{rytgaard2021estimation}.

\subsection{One-step targeted maximum likelihood estimation}
\label{sec:ulfm}

To construct the one-step targeted maximum likelihood estimation
procedure for the multivariate parameter
\(\Psi \, : \, \mathcal{M} \rightarrow \mathscr{H}_{\Sigma_d}\), we
construct a one-dimensional universal least favorable submodel
\(\lambda_{\eps}\) for
\(\lambda = (\lambda_l \, : \, l = 1, \ldots, J)\) such that for any
\(\eps\ge 0\),
\begin{align}
  \frac{d}{d\eps} \mathbb{P}_n \mathscr{L} ( \lambda_{\eps} ) = \Vert \mathbb{P}_n \tilde{D}^*(\lambda_{\eps},
  \pi, S^c)\Vert_{\Sigma_d} ,  
  \label{eq:ulfm:def}
\end{align}
for a loss function \((O, \lambda) \mapsto \mathscr{L}(\lambda) (O) \)
and fixed \(\pi, S^c\).

\begin{definition}[Universal least favorable submodel]
  For each cause-specific hazard, we define as follows: 
\begin{align}
  \lambda_{l, \eps}  (t) = \lambda_l (t) \exp
  \bigg( \int_0^{\eps} \frac{
  \big\langle \mathbb{P}_n \tilde{D}^*(\lambda_x,\pi,S^c) ,  \tilde{h}_{l, t} (\lambda_x,\pi,S^c)
  \rangle_{\Sigma_d}
  }{
  \Vert \mathbb{P}_n
  \tilde{D}^*(\lambda_x,\pi,S^c) \Vert_{\Sigma_d}
  } dx \bigg), \qquad l=1, \ldots, J,
  \label{eq:ulfm}
\end{align}
and refer to
\(\lambda_{\eps}= ( \lambda_{l, \eps} \, :\, l = 1,\ldots,
J) \) as the universal least favorable submodel.
\label{defi:ulfm}
\end{definition}

We show that Definition \ref{defi:ulfm} indeed defines a universal
least favorable submodel for the log-likelihood loss function. Let
\((O, \lambda) \mapsto \mathscr{L}(\lambda) (O) \) denote the sum loss
function given as
 \begin{align}
   \mathscr{L} (\lambda ) (O)  = - \sum_{l = 1}^J \bigg(
   \int_0^{\tau}\log \lambda_l(t \mid A, L)) N(dt)
   - \int_0^{\tau} \1 \lbrace \tilde{T} \ge t\rbrace \lambda_l(t \mid A, L) dt\bigg) .
   \label{eq:loss}
 \end{align}
 For this loss function and the universal least favorable submodel
 from Definition \ref{defi:ulfm} we have that
\begin{align*}
  \frac{d}{d\eps}     \mathbb{P}_n \mathscr{L} ( \lambda_{\eps} )
  &  =\frac{
    \big( \mathbb{P}_n \tilde{D}^*(\lambda_{\eps},\pi,S^c) \big)^\top \Sigma_d^{-1}
    }{
    \Vert \mathbb{P}_n \tilde{D}^*(\lambda_{\eps},\pi,S^c) \Vert_{\Sigma_d} 
    }   \mathbb{P}_n \bigg( \sum_{l=1}^J \int_0^{\tau} \tilde{h}_{l, t} (\lambda_\eps, \pi, S^c) \big( N_l(dt) -
    \lambda_{l,\eps}(t)dt\big) \bigg)\\
  &   = \frac{
    \big( \mathbb{P}_n \tilde{D}^*(\lambda_{\eps},\pi,S^c) \big)^\top \Sigma_d^{-1}
    }{
    \Vert \mathbb{P}_n \tilde{D}^*(\lambda_{\eps},\pi,S^c) \Vert_{\Sigma_d} 
    }  \mathbb{P}_n \tilde{D}^*(\lambda_{\eps},\pi,S^c) =     \Vert \mathbb{P}_n \tilde{D}^*(\lambda_{\eps},\pi,S^c) \Vert_{\Sigma_d} ,  
\end{align*}
which verifies the desired property \eqref{eq:ulfm:def}.
Particularly, the maximum likelihood estimator along the path defined
by the universal least favorable submodel,
\begin{align}
  \hat{\eps}_n = \underset{\eps \in \R}{\mathrm{argmin}} \,\, 
  \mathbb{P}_n \mathscr{L}( \lambda_{\eps} ),
  \label{eq:mle:ulfm}
\end{align}
is a local maximum and thus solves
\( \Vert \mathbb{P}_n \tilde{D}^*(\lambda_{\hat{\eps}_n}, \pi,
S^c)\Vert_{\Sigma_d} = 0\). It thus follows that
\( \mathbb{P}_n \tilde{D}_{j,t_k}^*(\lambda_{\hat{\eps}_n}, \pi,
S^c)=0\) for every \(j=1, \ldots, J\) and \(k=1, \ldots, K\), i.e.,
all desired score equations are solved.

\subsection{Additional steps for implementation}
\label{sec:implementation}
The universal least favorable submodel defined by Equation
\eqref{eq:ulfm} implies a recursive implementation of the one-step
targeting procedure. This follows the practical construction suggested
by \citet[][Section 5.5.1]{van2018targeted}. Particularly, for each
\(l=1,\ldots, J\), we define as follows
\begin{align*}
  \lambda_{l, dx}&= \lambda_{l} \exp\bigg(  \frac{
                      \big( \mathbb{P}_n \tilde{D}^*(\lambda_{}, \pi, S^c) \big)^\top \Sigma_{d}^{-1}
                      \tilde{h}_{l, t} (\lambda_{}, \pi, S^c)
                      }{
                      \Vert \mathbb{P}_n
                      \tilde{D}^*(\lambda_{}, \pi, S^c) \Vert_{\Sigma_d} 
                      }\bigg) ,
                      \intertext{and, for \(m\ge 1\), }
                      \lambda_{l, (m+1)dx} &= \lambda_{l,mdx}  \exp\bigg( \frac{
                                               \big( \mathbb{P}_n \tilde{D}^*(\lambda_{mdx}, \pi, S^c)  \big)^\top
                                               \Sigma^{-1} \tilde{h}_{l, t} (\lambda_{mdx}, \pi, S^c)
                                               }{
                                               \Vert \mathbb{P}_n
                                               \tilde{D}^*(\lambda_{mdx}, \pi, S^c) \Vert_{\Sigma_d} 
                                               }\bigg) , 
\end{align*}
with a small step size \(dx\). The maximum likelihood estimator
\eqref{eq:mle:ulfm} can now be found by recursively updating any
current estimator with the small step size \(dx\).  To see why this
work, we repeat the arguments of \citet[][Section
5.5.1]{van2018targeted}. First note that the universal least favorable
submodel \eqref{eq:ulfm} can be written as a product integral
\citep{gill1990survey,andersen2012statistical} as follows
\begin{align*}
  \lambda_{l, \eps}  (t) = \lambda_l (t) \Prodi_{0}^{\eps}
  \bigg(  1 +  \frac{
  \big( \mathbb{P}_n \tilde{D}^*(\lambda_x, \pi,S^c) \big)^\top \Sigma_d^{-1} \tilde{h}_{l, t} (\lambda_x, \pi,S^c)
  }{
  \Vert \mathbb{P}_n
 \tilde{D}^*(\lambda_x, \pi,S^c) \Vert_{\Sigma_d}
  } dx \bigg), \qquad l=1, \ldots, J,
\end{align*}
such that
\begin{align*}
  \lambda_{l,dx}^*  (t) =  \lambda_l (t)
  \bigg(  1 +  \frac{
  \big( \mathbb{P}_n \tilde{D}^*(\lambda_x, \pi,S^c) \big)^\top \Sigma_d^{-1} \tilde{h}_{l, t} (\lambda_x, \pi,S^c)
  }{
  \Vert \mathbb{P}_n
  \tilde{D}^*(\lambda_x, \pi,S^c) \Vert_{\Sigma_d}
  } dx \bigg), \qquad l=1, \ldots, J.
\end{align*}
This corresponds \citep[see][Section 5.5.1, p. 66]{van2018targeted} to
a Taylor expansion of a local least favorable submodel
\(\lambda_\delta^{\mathrm{LLFM}}\) with
\begin{align*}
  \bigg\langle \frac{d}{d\delta}\bigg\vert_{\delta =0} \mathscr{L} (\lambda_\delta^{\mathrm{LLFM}})
  , \delta \bigg\rangle_{\Sigma_d} =
  \langle D^*(P), \delta \rangle_{\Sigma_d} . 
\end{align*}
We maximize the score locally over \(\delta\) with
\(\Vert \delta \Vert \le dx\) using the local least favorable
submodel.  Maximizing
\(\delta \mapsto \mathbb{P}_n \mathscr{L}
(\lambda_\delta^{\mathrm{LLFM}})\) corresponds to maximizing
\(\delta \mapsto \langle \mathbb{P}_n D^*(P), \delta
\rangle_{\Sigma_d}\), which, by the Cauchy-Schwartz inequality, is
maximized by
\begin{align*}
\delta^* = \frac{\mathbb{P}_n D^*(P)}{ \Vert \mathbb{P}_n D^*(P) \Vert_{\Sigma_d} } dx.
\end{align*}

Following up on Remark \ref{rem:choice:Sigma:euclidean:norm}, we
finish this section proposing the following two other choices for the
matrix \(\Sigma_d \in \R^{d} \times \R^{d}\) in the definition of the
Hilbert space norm.

{
  \begin{remark}[Norm weighted by the variance of the efficient
    influence function]
    Let \(\Sigma_d \in \R^{d} \times \R^{d}\) be the diagonal matrix
    with diagonal values given by the estimated variances
    \(\hat{\sigma}_{j,k}^2 = \mathbb{P}_n
    (\tilde{D}_{j,k}^*(\hat{\lambda}_{\eps_n},\hat{\pi}_n,\hat{S}^c_n))^2\)
    of the efficient influence functions.
  \label{rem:choice:sigma:variance:norm}
\end{remark}
}

{
\begin{remark}[Norm weighted by the covariance of the efficient
  influence function]
  Let \(\Sigma_d \in \R^{d} \times \R^{d}\) be the empirical
  covariance matrix
  \(\Sigma_{n} = \mathbb{P}_n
  D^*(\hat{\lambda}_{\eps_n},\hat{\pi}_n,\hat{S}^c_n)
  D^*(\hat{\lambda}_{\eps_n},\hat{\pi}_n,\hat{S}^c_n)^\top\) of the
  stacked efficient influence function
  \(\tilde{D}^*(\hat{\lambda}_{\eps_n},\hat{\pi}_n,\hat{S}^c_n)\).
  \label{rem:choice:Sigma:covariance:norm}
\end{remark}
}

\section{Infinite-dimensional target parameters}
\label{sec:infinite:dimensional}

We now turn our attention to the estimation of infinite-dimensional
target parameters \(\Psi \, : \, \mathcal{M} \rightarrow \mathscr{H}\)
with values in a Hilbert space \(\mathscr{H}\) of real-valued
functions on \(\R_+\) endowed with the inner product and corresponding
norm 
\begin{align}
  \langle f_1, f_2 \rangle = \int f_1(t) f_2(t) d\Gamma (t) , \qquad
  \Vert f\Vert = \sqrt{\langle f, f \rangle }
  \label{eq:define:inner:product:2}
\end{align}
for a user-supplied positive and finite measure \(\Gamma\). We
consider particularly the case where \(\Psi(P)\) is the
treatment-specific average cause one specific absolute risk curve
across time:
\begin{align*}
  \Psi(P)(t) = \EE \big[ P(T \le t, \Delta = 1 \mid A=a^*, L)\big]
  = \EE \big[ F_1( t \mid A=a^*, L)\big], \qquad t \in [ 0, \tau] . 
\end{align*}
When there are no competing risks, this is simply one minus the
treatment-specific survival curve. As in Section \ref{ssec:eff:ic}, we
let \(D^*_{t}(P)\) denote the relevant part of the efficient influence
function for \( \Psi (P)(t) \), for \(t\ge 0\), i.e.,
\begin{align}
  D_{t}^*(P) (O) =  \sum_{l=1}^J \int_0^\tau   h_{1,l,t,s} (P) (O) \, \big( N_l(ds) -
  \1 \lbrace \tilde{T}\ge s\rbrace \lambda_l(s \, \vert \, A,L) ds \big) , 
\end{align}
with the functions \( h_{1,l,t,s} \) for \(l=1, \ldots, J\)
defined by
\begin{align*}
  h_{1,l,t,s} (P)(O) =
  \frac{
  \1 \lbrace A=a^*\rbrace }{ \pi (A \, \vert \, L) } \frac{\1 \lbrace s  \le t \rbrace}{
       S^c( s- \, \vert \, A,L)
  } \begin{cases}
    1- \frac{ F_1 ( t \mid A, L) - F_1( s\mid A, L)}{
      S(s \, \vert \, A,L)}, &\text{when }\, l=1, \\
    -\frac{ F_1 ( t  \mid A, L) - F_1 ( s \mid A, L)}{
      S(s \, \vert \, A,L) },  &\text{when }\, l \neq 1.
  \end{cases}
\end{align*}
Again we use the alternative notation \(\tilde{\Psi}(\lambda)\),
\( \tilde{D}^* (\lambda, \pi, S^c) (O) \) and
\(\tilde{h}_{l, t} (\lambda,\pi,S^c)\) when appropriate.  The
universal least favorable submodel is defined according to Definition
\ref{defi:ulfm} with \(\langle \cdot, \cdot\rangle\) substituted for
\(\langle \cdot, \cdot\rangle_{\Sigma_d}\):
\begin{align*}
  \lambda_{l, \eps}  (t) = \lambda_l (t) \exp
  \bigg( \int_0^{\eps} \frac{
  \langle \mathbb{P}_n \tilde{D}^*(\lambda_x,\pi,S^c),  \tilde{h}_{l, t}
  (\lambda_x,\pi,S^c)\rangle
  }{
  \Vert \mathbb{P}_n
  \tilde{D}^*(\lambda_x,\pi,S^c) \Vert 
  } dx \bigg), \quad l=1, \ldots, J, 
\end{align*}
such that, together with
\(\lambda_{\eps} = (\lambda_{l,\eps}\, : \, l =1, \ldots, J)\) and the
sum loss function \(\mathscr{L}\) from Section \ref{sec:ulfm}, we have
that
\begin{align*}
  \frac{d}{d\eps}    \mathbb{P}_n \mathscr{L} ( \lambda_{ \eps})
  &  =\frac{
    \int  \mathbb{P}_n \tilde{D}_t^*(\lambda_\eps, \pi, S)  \mathbb{P}_n \big( \sum_{l=1}^J
    \int \tilde{h}_{l, t} (\lambda_\eps, \pi, S) \big( N_l(dt) -
    \lambda_{l,\eps}(t)dt\big) \big) d\Gamma (t) 
    }{
    \Vert \mathbb{P}_n \tilde{D}^*(\lambda_{\eps},\pi,S^c) \Vert 
    }
  \\
  &   = \frac{
    \int  \mathbb{P}_n \tilde{D}_t^*(\lambda_\eps, \pi, S)  \mathbb{P}_n
    \tilde{D}_t^*(\lambda_\eps, \pi, S) d\Gamma (t) 
    }{
    \Vert \mathbb{P}_n \tilde{D}_t^*(\lambda_\eps, \pi, S) \Vert 
    }   =     \Vert \mathbb{P}_n \tilde{D}_t^*(\lambda_\eps, \pi, S) \Vert .
\end{align*}

\begin{remark}[Implementation for an infinite-dimensional parameter]
  As we will see in Section \ref{sec:grid:for:sup}, the practical
  construction of the one-step targeting procedure can in fact be
  carried out along a grid of time-points with a grid that is chosen
  fine enough. Particularly, this means that the implementation of
  one-step targeting may follow exactly along that of Section
  \ref{sec:implementation}.
\label{rem:practical:infinite}
\end{remark}

\subsection{Conditions for weak convergence}
\label{sec:conditions:inference:infinite}

We here review what is needed for simultaneous inference for the
infinite-dimensional target parameter. These conditions are also
presented in \citet[][Chapter 5.6]{van2018targeted}.  Let
\(\hat{\psi}_{n}^*(t) = \Psi(\hat{P}^*_n) (t)\) denote an estimator
for \(\psi_0(t) = \Psi (P_0)(t)\) and let
\(\hat{\psi}_{n}^* = (\hat{\psi}_{n}^* (t) \, : \, t \in \R_+)\) as
well as \(\psi_0 = (\psi_0 (t) \, : \, t \in \R_+)\). Further define
the second-order remainder
\(R_2( \hat{P}_n^*, P_0) = (R_{2,t}( \hat{P}_n^*, P_0) \, : \, t\in
\R_+) \) by
\begin{align*}
 R_2( \hat{P}_n^*, P_0) =  \hat{\psi}^*_n - \psi_0  +  P_0 D^*(\hat{P}^*_n).
\end{align*}
Asymptotic linearity and efficiency for \(\hat{\psi}_{n}^* (t) \) is
established for fixed \(j\) and \(k\) by the usual conditions
\citep{rytgaard2021estimation}. The below conditions now tell us what we need
for asymptotic efficiency in supremum norm of
\(\hat{\psi}_{n}^* =\Psi(\hat{P}^*_n)\) in addition to the pointwise
efficiency:
\begin{enumerate}
\item[(i)] The estimator solves the efficient influence curve equation
  across all time-points, i.e.,
  \begin{align}
    \sup_{t \in [0,\tau]} \,\mathbb{P}_n D_t^*(\hat{P}^*_n) = o_P(n^{-1/2})
    ;
    \label{eq:eic:eq}
  \end{align}
\item[(ii)]  \(\lbrace D_t^*(P) \, : \, P \in \mathcal{M}, \, t \in
    \R_+\rbrace\) is a \(P_0\)-Donsker class and
    \(\sup_{t\in [0,\tau]} P_0 \big( D^*_{t}(\hat{P}^*_n) - D^*_{t}(P_0)
    \big)^2\rightarrow 0\) in probability, and
  \item[(iii)]
    \(\sup_{t\in [0,\tau]} \vert R_{2,t}(\hat{P}^*_n, P_0) \vert =
    o_P(n^{-1/2})\).
  \end{enumerate}
  Indeed, under conditions (i)--(iii) we have that
   \begin{align*}
    \sqrt{n} ( \hat{\psi}^*_n  - \psi_0)
    = \sqrt{n}(\mathbb{P}_n - P_0) D^*(P_0) + o_P(1)
    \overset{\mathcal{D}}{\rightarrow} \mathbb{G}_0,
  \end{align*}
  where \(\overset{\mathcal{D}}{\rightarrow} \) denotes convergence in
  distribution and \( \mathbb{G}_0\) is a Gaussian process with
  covariance structure given by the covariance function
  \(\rho (t_1, t_2) = P_0 D^*_{t_1} (P_0) D^*_{t_2} (P_0)\); that is,
  \( \sqrt{n}( \hat{\psi}^*_n - \psi_0)\) converges weakly as a random
  element of the \cadlag function space endowed with the supremum norm
  to \(\mathbb{G}_0\).

  \begin{remark}[Conditions (ii)--(iii)]
    Conditions (ii) and (iii) are covered by Assumption
    \ref{ass:M:cadlag}. Indeed, the class of \cadlag functions with
    finite variation is a well-known Donsker class \citep{van1996weak}
    and since the efficient influence function is a well-behaved
    mapping of the nuisance parameters, it inherits the Donsker
    properties. Moreover, Assumption \ref{ass:M:cadlag} allows for the
    construction of highly adaptive lasso estimators
    \citep{van2017generally} that have been shown to converge at a
    rate faster than \(n^{-1/3-\eta}\) with respect to the
    Kullback-Leibler dissimilarity, for any \(\eta>0\), to the true
    function \citep{2019arXiv190709244B}. That this is enough to
    establish (iii) follows from the double robustness structure of
    the second-order remainder, see \citet[][Supplementary
    Material]{rytgaard2021estimation}.
  \end{remark}

\subsection{Targeting over a grid to achieve supremum norm inference}
\label{sec:grid:for:sup}

In Remark \ref{rem:practical:infinite} we claimed that we only need to
do the targeting over a grid of time-points to solve all score
equations, i.e., to solve Equation \eqref{eq:eic:eq} where
\(\hat{P}^*_n\) now denotes the one-step targeted estimator. As
presented in Section \ref{sec:conditions:inference:infinite}, weak
convergence and asymptotic efficiency in the supremum norm then
follows. The claim of Remark \ref{rem:practical:infinite} follows
under the conditions given by the following theorem.

\begin{theorem}[Targeting over a grid]
  Define a grid \( 0 \le t_1 < t_2 < \cdots < t_{M_n} \le \tau \) of
  time-points in \([0,\tau]\) fine enough such that
  \(\max_m \, (t_{m}-t_{m-1}) = O_P(n^{-1/3- \eta})\) for some
  \(\eta >0\). For each \( t\in [0,\tau]\) let
  \(m(t) := \min_m \vert t -t_m \vert\). Assume, for all
  \( t\in [0,\tau]\) that:
  \begin{itemize}
  \item[A1.]
    \( P_0 \big( D^*_{t}(\hat{P}^*_n) - D^*_{t_{m(t)}}(\hat{P}^*_n)
    \big)^2\overset{P}{\rightarrow} 0\); 
  \item[A2.]
    \(\big\Vert ( h_{1,\ell,t_{m(t)} } - h_{1,\ell,t } ) (\hat{P}^*_n)
    \big\Vert_{\mu_0 \otimes \pi_0 \otimes \rho} \le K' \vert t_{m(t)}
    -t \vert^{1/2}\) for a constant \(K'>0\); here,
    \( \Vert \cdot \Vert_{\mu_0 \otimes \pi_0 \otimes \rho} \) denotes
    the \(L_2( {\mu_0 \otimes \pi_0 \otimes \rho} )\)-norm where
    \(\rho\) the Lebesgue measure and
    \({\mu_0 \otimes \pi_0 \otimes \rho} \) is the product measure of
    \(\mu_0\), \(\pi_0\) and \(\rho\).
  \end{itemize} Now, if
  \begin{align} \max_{t\in \lbrace t_1, \ldots, t_{M_n}\rbrace} \,
    \mathbb{P}_n D_t^*(\hat{P}^*_n) = o_P(n^{-1/2}) ,
  \label{eq:eic:eq:max}
\end{align}
then we have that
\begin{align}
  \sup_{t\in [0, \tau]} \, \mathbb{P}_n D_t^*(\hat{P}^*_n) = o_P(n^{-1/2}) . 
  \label{eq:eic:eq:sup}
\end{align}
\label{thm:targeting:grid}
\end{theorem}

\begin{proof} See Appendix A.
\end{proof}

\begin{remark}[Assumption A2] In Appendix B we verify that Assumption
  A2 holds under a Lipschipz type condition on the cause one
  subdistribution.  We further note that the exponent of
  \(\alpha= 1/2\) for the upper bound \(\vert t_m - t\vert^{\alpha}\)
  of Assumption A2 is specific for our parameter of interest (the
  absolute risk function). For other choices of target parameters, it
  may be that \(\alpha \in [\tfrac{1}{2}, 1]\) such that a coarser
  grid is really needed for weak convergence.  \label{remark:A2}
\end{remark}
  
\subsection{Final remarks on the implementations}
\label{sec:practical:implementation:infinite}

We carry out the targeting over the grid
\( 0 \le t_1 < t_2 < \cdots < t_{M_n} \le \tau \) according to the
procedure outlined in Section \ref{sec:grid:for:sup}. Particularly,
this yields that
\begin{align*}
  \Vert \mathbb{P}_n \tilde{D}^*(\hat{\lambda}_{\hat{\eps}_n}, \hat{\pi}_n,
  \hat{S}^c_n)\Vert_{\Sigma_{M_n}} = o_P(n^{-1/2}),
\end{align*}
where \(\Vert \cdot \Vert_{\Sigma_{M_n}}\) denotes the
grid-\(M_n\)-dependent Hilbert space norm.  To apply Theorem
\ref{thm:targeting:grid} we need to control the maximum norm:
Specifically
\begin{align*}
  \max_{t\in \lbrace t_1, \ldots, t_{M_n}\rbrace} \,
  \frac{ \mathbb{P}_n \tilde{D}_t^*(\hat{\lambda}_{\eps}, \hat{\pi}_n, \hat{S}^c_n
  )}{\hat{\sigma}_t} \le \frac{1}{\sqrt{n} \log n } , 
\end{align*}
where
\(\hat{\sigma}_t^2 = \mathbb{P}_n \tilde{D}_t^*(\hat{\lambda}_{\eps},
\hat{\pi}_n, \hat{S}^c_n )^2\), is enough to ensure
\eqref{eq:eic:eq:max} of Theorem \ref{thm:targeting:grid}. Lemma
\ref{lem:from:l2:to:sup} below tells us that we can achieve exactly
this.

\begin{lemma}
  Assume that the likelihood is bounded along the path defined by the
  universal least favorable model:
  \begin{align*}
    \sup_{\eps \ge 0 } \, \mathbb{P}_n \mathscr{L} ( \hat{\lambda}_{ \eps}) <\infty.
  \end{align*}
  Then we can choose \(\eps^*_n\) large enough so that for all
  \(\eps \ge \eps^*_n\) we have that
\begin{align*}
  \max_{t\in \lbrace t_1, \ldots, t_{M_n}\rbrace} \,
  \frac{ \mathbb{P}_n \tilde{D}_t^*(\hat{\lambda}_{\eps}, \hat{\pi}_n, \hat{S}^c_n
  )}{\sigma_t} \le \frac{1}{\sqrt{n} \log n } . 
\end{align*}
\label{lem:from:l2:to:sup}
\end{lemma}

\begin{proof}
  Assume, for contradiction, that
  \(\lim_{\eps \rightarrow \infty} \Vert \mathbb{P}_n
  \tilde{D}_t^*(\hat{\lambda}_{\eps}, \hat{\pi}_n, \hat{S}^c_n )
  \Vert_{\Sigma_{M_n}} > \eta'' \) for some \(\eta''>0\). Since the
  likelihood is an increasing function which is bounded and since
  \(\tfrac{d}{d\eps}\,\mathbb{P}_n \mathscr{L} ( \hat{\lambda}_{
    \eps}) = \Vert \mathbb{P}_n \tilde{D}_t^*(\hat{\lambda}_{\eps},
  \hat{\pi}_n, \hat{S}^c_n ) \Vert_{\Sigma_{M_n}} \), we reach a
  contradiction. Thus we have that
  \(\lim_{\eps \rightarrow \infty} \Vert \mathbb{P}_n
  \tilde{D}_t^*(\hat{\lambda}_{\eps}, \hat{\pi}_n, \hat{S}^c_n )
  \Vert_{\Sigma_{M_n}} = 0 \), and, accordingly, we can find an
  \(\hat{\eps}_n^*\) such that for all \( \eps \ge \hat{\eps}_n^*\) we
  have
  \(\Vert \mathbb{P}_n \tilde{D}_t^*(\hat{\lambda}_{\eps},
  \hat{\pi}_n, \hat{S}^c_n ) \Vert_{\Sigma_{M_n}} \le s_n\) for any
  choice \(s_n>0\).  The claim now follows since
  \begin{align}
    \max_{t\in \lbrace t_1, \ldots, t_{M_n} \rbrace} \, \frac{\mathbb{P}_n
    \tilde{D}_t^*(\hat{\lambda}_{\eps}, \hat{\pi}_n, \hat{S}^c_n
    )}{\sigma_{t}} \le c_n \, \Vert \mathbb{P}_n\tilde{D}_t^*(\hat{\lambda}_{\eps}, \hat{\pi}_n, \hat{S}^c_n
    )
    \Vert_{\Sigma_{M_n}}, \label{eq:max:bounded:by:norm}
  \end{align}
  where \(c_n > 0\) is a constant depending on \(n\). Thus, if we can
  control the right hand side of \eqref{eq:max:bounded:by:norm}, we
  can control the left hand side.
\end{proof}\color{black}

That the bound \eqref{eq:max:bounded:by:norm} holds for the
variance-weighted Hilbert space norm proposed in Remark
\ref{rem:choice:sigma:variance:norm} is straightforward. In Appendix C
we further demonstrate that it holds for the covariance-weighted norm
proposed in Remark \ref{rem:choice:Sigma:covariance:norm} with
\(c_n = \sqrt{ M_n}\).

\section{Empirical study}
\label{sec:empirical:study}

Our empirical study consists of a demonstration of our proposed
methodology on a publicly available dataset and further a simulation
study for proof of concept. The purpose is to demonstrate the
theoretical properties and to explore the properties of different
variations of targeting.

Specifically, we compare throughout the results using the one-step
estimator to the iterative counterpart of
\cite{rytgaard2021estimation} which can be used to target
one-dimensional target parameters; thus, to apply this method, we
target each component of the particular multivariate target parameter
separately.

\subsection{Demonstration in a dataset}
\label{sec:pbc:demo}

For a simple demonstration of our methods, we consider the publicly
available dataset from the Mayo Clinic trial in primary biliary
cholangitis (earlier called primary biliary cirrhosis, or, short, PBC)
conducted between 1974 and 1984 as available from the survival package
\citep{survivalpackage} in \texttt{R}. We also note that
an almost identical version of the dataset is described in
\citet[][Appendix D]{fleming2011counting}. The trial included
\(n=312\) patients who were randomized to treatment with the drug
D-penicillamine (\(A=1\)) or to placebo (\(A=0\)). The patients were
followed over time until the earliest of liver transplantation
(\(\Delta = 1\)), death (\(\Delta = 2\)) or end of follow-up. In
previous analyses, interest was in the effect of the treatment on the
risk only of cause one events (liver transplantation), whereas we here
consider estimation of the treatment-specific absolute risks of both
causes. Our point is solely to demonstrate that our one-step approach
leads to a compatible estimator, respecting the bounds of the
parameter space, whereas any approach estimating each real-valued
component separately is not guaranteed to do so.

We define our target parameter as the vector of treatment-specific
absolute risks for both event types across the ten time-points, i.e.,
\begin{align*}
  \Psi(P) = \big( \Psi_{1,1}(P) ,\ldots, \Psi_{1,10}(P), \Psi_{2,1}(P) , \ldots, \Psi_{2,10}(P)\big),
\end{align*}
with \( \Psi_{1,k}(P) = \EE \big[ F_j ( t_k \mid A=1, L)\big]\) for
\(j=1,2\) and \(k=1, \ldots, 10\). Since for all \(t\) we have that
\(S ( t\mid A, L) = 1 - F_1 (t \mid A,L) - F_2(t \mid A,L)\), and thus
\(D^*_{S,t} = -(D^*_{1,t} + D^*_{2,t}) \), the one-step targeted
estimator which solves the score equations for the two
treatment-specific absolute risk functions necessarily solves the
score equation for the treatment-specific survival probability
too. Table \ref{tab:pbc:one} presents the results using the one-step
algorithm: As can be noted, the sum of all three estimated state
probabilities (\(\hat{F}_1+\hat{F}_2 + \hat{S}\)) is ensured to be 1
at all time-points.  Table \ref{tab:pbc:iter} shows the results using
the iterative targeting procedure of \cite{rytgaard2021estimation} where each
treatment-specific probability is targeted on its own leading to
incompatible estimators not guaranteed summing up to 1.

\begin{table}[ht]
\centering
\begin{adjustbox}{tabular=rcccccccccc,center}
  & $t_{1}$ & $t_{2}$ & $t_{3}$ & $t_{4}$ & $t_{5}$ & $t_{6}$ & $t_{7}$ & $t_{8}$ & $t_{9}$ & $t_{10}$ \\
  \hline\hline
  $\hat{F}_1$ & 0.0000 & 0.0051 & 0.0313 & 0.0317 & 0.0472 & 0.0472 & 0.0557 & 0.0679 & 0.0822 & 0.0822 \\
  $\hat{F}_2$ & 0.0623 & 0.0743 & 0.1157 & 0.1871 & 0.2234 & 0.2741 & 0.3111 & 0.3726 & 0.4183 & 0.4185 \\
  $\hat{S}$ & 0.9377 & 0.9206 & 0.8530 & 0.7812 & 0.7294 & 0.6787 & 0.6332 & 0.5596 & 0.4996 & 0.4993 \\
  \hline
  sum & 1.0000 & 1.0000 & 1.0000 & 1.0000 & 1.0000 & 1.0000 & 1.0000 & 1.0000 & 1.0000 & 1.0000 \\
\end{adjustbox}
\caption{Estimated treatment-specific state occupation probabilities
  using the one-step algorithm for the dataset from the Mayo Clinic
  trial. }
\label{tab:pbc:one}
\end{table}

\begin{table}[ht]
\centering
\begin{adjustbox}{tabular=rcccccccccc,center}
  & $t_{1}$ & $t_{2}$ & $t_{3}$ & $t_{4}$ & $t_{5}$ & $t_{6}$ & $t_{7}$ & $t_{8}$ & $t_{9}$ & $t_{10}$ \\
  \hline\hline
  $\hat{F}_1$ & 0.0000 & 0.0051 & 0.0305 & 0.0307 & 0.0449 & 0.0449 & 0.0532 & 0.0665 & 0.0821 & 0.0821 \\
  $\hat{F}_2$ & 0.0674 & 0.0792 & 0.1203 & 0.2065 & 0.2443 & 0.2977 & 0.3329 & 0.3969 & 0.4450 & 0.4452 \\
  $\hat{S}$ & 0.9326 & 0.9156 & 0.8489 & 0.7628 & 0.7098 & 0.6574 & 0.6151 & 0.5414 & 0.4811 & 0.4785 \\
  \hline
  sum & 1.0000 & 1.0000 & 0.9997 & 1.0000 & 0.9989 & 0.9999 & 1.0013 & 1.0047 & 1.0082 & 1.0058 \\
\end{adjustbox}
\caption{Estimated treatment-specific state occupation probabilities
  using the iterative algorithm for the dataset from the Mayo Clinic
  trial.}
\label{tab:pbc:iter}
\end{table}

\subsection{Simulation study with survival outcome}
\label{sec:sim:survival}

For proof of concept, we consider a simulation study with just a
single cause of interest. We draw three baseline covariates
\(L=(L_1,L_2,L_3)\) such that \(L_1\) is uniform on \([-1,1]\) and
\(L_2,L_3\) are uniform on \([0,1]\). We let treatment be randomized
and censoring be covariate independent. The hazard for distribution of
the survival time is given as follows
\begin{align*}
  \lambda_1 ( t \mid A, L) &= \lambda_{0}(t) \exp (-0.15 A  + 1.2 L_1^2 ), 
\end{align*}
with the baseline hazard corresponding to a Weibull distribution. We
consider first estimation of a multivariate target parameter, defined
as the vector of average treatment effects on the survival curve
evaluated across ten pre-specified time-points, i.e.,
\begin{align*}
  \Psi(P) = \big( \Psi_{1}(P) ,
  \ldots,\Psi_{10}(P)\big),
\end{align*}
with
\( \Psi_{k}(P) = \EE \big[ S ( t_k \mid A=1, L) - S ( t_k \mid A=0,
L)\big]\) for \(k=1,\ldots, 10\).

To construct the one-step estimator, we repeat the updating steps
described in Section \ref{sec:implementation} until
\begin{align*}
  \max_{ t\in \lbrace t_1, \ldots, t_{10}\rbrace} \,
  \frac{ \mathbb{P}_n \tilde{D}_{t}^*(\hat{\lambda}_{\eps^*}, \hat{\pi}_n,
  \hat{S}^c_n )}{\hat{\sigma}_t} \le \frac{1}{\sqrt{n} \log n } ,
\end{align*}
at which point, particularly,
\( \mathbb{P}_n \tilde{D}^*_{t} (\hat{\lambda}_{\eps^*}, \hat{\pi}_n,
\hat{S}^c_n) = o_P(n^{-1/2}) \) for \( k=1,\ldots, 10\). The
corresponding estimator for the target parameter is
\(\hat{\psi}^* = ( \hat{\psi}_{1}^*, \ldots, \hat{\psi}_{10}^*)= \big(
\tilde{\Psi}_{1} (\hat{\lambda}_{\eps^*}), \ldots,
\tilde{\Psi}_{10}(\hat{\lambda}_{\eps^*})\big)\).  Since all efficient
score equations are solved simultaneously, asymptotic linearity of the
estimators applies across time-points \(t_k\), and we have
\begin{align}
  \sqrt{n}
  ( \hat{\psi}^* - \psi_0 )
  \overset{\mathcal{D}}{\longrightarrow} N( 0 , \Sigma_0) ,
  \label{eq:multivariate:inference}
\end{align}
where \(\psi_0 = (\Psi_{1}( P_0), \ldots, \Psi_{10}( P_0))\) and
\(\Sigma_0 = D^* (P_0)^\top D^* (P_0) \,\,\in \R^{d}\times \R^{d}\) is
the covariance matrix of the stacked efficient influence
function. Particularly, the asymptotic distribution in
\eqref{eq:multivariate:inference} can be used to provide simultaneous
confidence intervals
\( \big( \hat{\psi}_{k}^* - \tilde{q}_{0.95} {\hat{\sigma}_{t_k}} /
{\sqrt{n}}, \hat{\psi}_{k}^* + \tilde{q}_{0.95} {\hat{\sigma}_{t_k}}
/{\sqrt{n}} \big)\), where
\(\hat{\sigma}_{t_k}^2 = \mathbb{P}_n (\tilde{D}^*_{t_k} (
\hat{\lambda}_{\eps^*}, \hat{\pi}_n, \hat{S}^c_n))^2\) estimates the
variance of the \(k\)th efficient influence function, and
\(\tilde{q}_{0.95}\) is the 95\% quantile for the distribution of
\(\max_{k} \vert \hat{\psi}_{k}^* - \Psi_{k}(P_0)\vert / (\sigma_{k} /
\sqrt{n})\). Figure \ref{fig:sim:ci:cov:marg:sim:survival} illustrates
confidence intervals for a single simulated data set, and further
shows the coverage across simulation repetitions of marginal and
simultaneous confidence intervals based on the efficient influence
function for the \(\sigma_n\)-weighted one-step estimator.  Table
\ref{table:surv:1:rel:mse} further shows the relative mean squared
error for the \(\sigma_n\)-weighted one-step estimator compared to the
unweighted one-step estimator, the $\Sigma_n$-weighted one-step
estimator, and the iteratively targeted estimator from
\cite{rytgaard2021estimation}.

\begin{table}[ht]
\centering
\begin{adjustbox}{tabular=rcccccccccc,center}
  \(t_k\) & 0.1 & 0.256 & 0.411 & 0.567 & 0.722 & 0.878 & 1.033 & 1.189 & 1.344 & 1.5 \\
  \hline\hline
  $\sigma_n$-weighted & 1.000 & 1.000 & 1.000 & 1.000 & 1.000 & 1.000 & 1.000 & 1.000 & 1.000 & 1.000 \\
  $\Sigma_n$-weighted & 1.256 & 1.047 & 1.038 & 1.001 & 1.001 & 1.011 & 1.009 & 1.030 & 1.014 & 1.000 \\
  unweighted & 0.980 & 0.997 & 0.993 & 0.981 & 0.982 & 0.990 & 0.989 & 0.988 & 0.986 & 0.976 \\
  iterative & 1.254 & 1.045 & 1.016 & 0.995 & 1.004 & 1.011 & 1.008 & 1.025 & 1.024 & 1.014 \\
\end{adjustbox}
\caption{Results from a simulation study with sample size \(n=200\)
  and \(M=500\) repetitions. Shown are the relative mean squared
  errors across the \(M=500\) simulation repetitions for the
  \(\sigma_n\)-weighted one-step estimator, the $\Sigma_n$-weighted
  one-step estimator, the unweighted one-step estimator and the
  iteratively targeted estimator from \cite{rytgaard2021estimation}
  all relative to the \(\sigma_n\)-weighted one-step estimator.  }
\label{table:surv:1:rel:mse}
\end{table}

\begin{figure}  
  \begin{center}   
  \makebox[\textwidth][c]{
    \includegraphics[width=1\textwidth,angle=0]{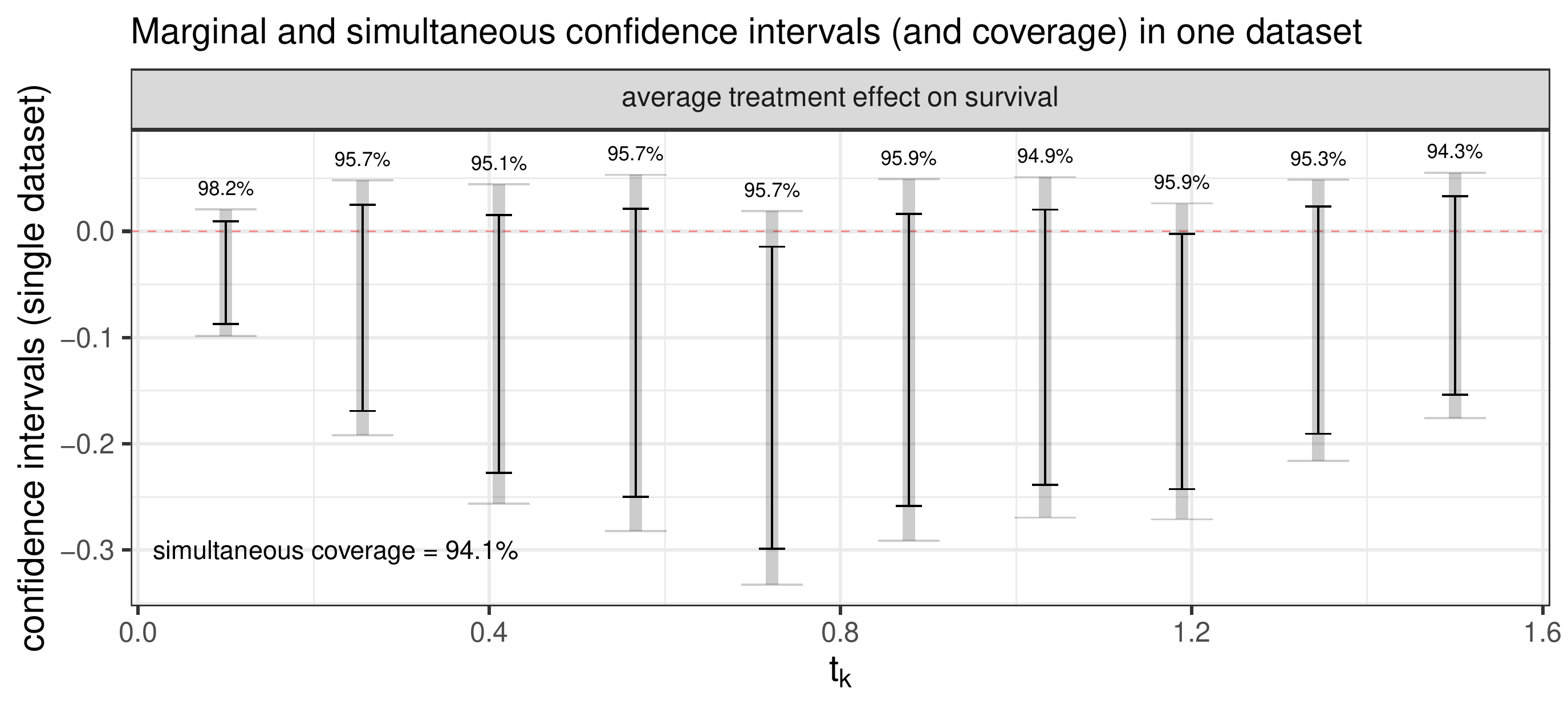}}
  \caption{Shown are marginal (black) and simultaneous (gray) \(95\%\)
    confidence intervals based on the efficient influence function for
    the \(\sigma_n\)-weighted one-step estimator in a single simulated
    data set (sample size \(n=200\)); above the bars are further shown
    marginal coverage across the \(M=500\) simulation repetitions. For
    this dataset, the overall hypothesis that there is no treatment
    effect on survival across the ten time-points is rejected. }
\label{fig:sim:ci:cov:marg:sim:survival} 
\end{center}  
\end{figure}

To further investigate empirically our claim of Theorem
\ref{thm:targeting:grid}, we next consider estimation of the
treatment-specific survival \textit{curve},
\begin{align*}
  \Psi(P)(t) = \EE \big[ P(T > t \mid A=1, L)\big]
  = 1- \EE \big[ F( t \mid A=1, L)\big], 
\end{align*}
across \(t \in [ 0, \tau]\). According to Theorem
\ref{thm:targeting:grid}, we only need to solve the efficient
influence curve equation over a grid of time-points to solve the
efficient influence curve equation across all times
\(t\in [ 0, \tau]\), as long as this grid is chosen fine enough. To
check this, we consider targeting over grids of increasing size while
checking if we have solved the efficient influence curve equation
across 100 randomly sampled time-points. In fact, for practical usage,
we suggest that the investigator does the same thing for their data;
indeed, one may choose the grid for a particular dataset by increasing
the size of the grid until the efficient influence curve equation is
solved up to a sufficient level for 100 randomly sampled
time-points. Table \ref{tab:fraction:solved} collects our experiences
with sample size \(n=1,000\) and a single simulation repetition. To
produce this table, we used the \(\Sigma_n\)-weighted one-step
estimator as we found this to converge considerably faster than the
other choices of norms.

\begin{table}[ht]
\centering
\begin{tabular}{rccccccc}
  \hline
  grid size used: & 20 & 30 & 40 & 60 & 80 & 100 & 120 \\ 
  fraction solved: & 0.72 & 0.88 & 0.88 & 0.88 & 0.95 & 0.98 & 0.99 \\ 
  worst-case ratio: & 3.00 & 2.21 & 1.84 & 2.09 & 1.94 & 1.26 & 1.47 \\ 
  \hline
\end{tabular}
\caption{Results from targeting over a grid of varying size. The first
  row shows the grid size used, the second row shows the fraction of
  the 100 randomly sampled time-points where the efficient influence
  curve equation is solved and the third row show the maximal ratio of
  the absolute value standardized empirical mean of the efficient
  influence curve
  \(\vert \mathbb{P}_n \tilde{D}^*_t(\hat{\lambda}_{\eps^*},
  \hat{\pi}_n,\hat{S}_n^c)\vert / \hat{\sigma}_t\), across the 100
  randomly sampled time-points, relative to the criterion
  \(1/(\sqrt{n}\log n)\).}
\label{tab:fraction:solved}
\end{table}

\subsection{Simulation study with competing risks}
\label{sec:sim:competing:risks}

To investigate the properties of the different choices of Hilbert
space norms (see Remarks
\ref{rem:choice:Sigma:euclidean:norm}--\ref{rem:choice:Sigma:covariance:norm})
and to further compare with separate estimation via the iteratively
targeted estimator from \cite{rytgaard2021estimation}, we consider
here a simulation study with two causes of interest.  The two
cause-specific hazards are given as follows
\begin{align*}
  \lambda_1 ( t \mid A, L) &= \lambda_{0}(t) \exp (-0.15 A  + 1.2 L_1^2 ), \\
  \lambda_2( t \mid A, L) &= \lambda_{0}(t) \exp (0.4 + 0.7L_1 - 0.4 A), 
\end{align*}
with the baseline hazard corresponding to a Weibull distribution.

We define now our target parameter as the vector of average treatment
effects on the absolute risks across three pre-specified time-points,
i.e.,
\begin{align*}
  \Psi(P) = \big( \Psi_{1,1}(P) ,
  \Psi_{1,2}(P), \Psi_{1,3}(P), \Psi_{2,1}(P) , \Psi_{2,2}(P), \Psi_{2,3}(P)\big),
\end{align*}
with
\( \Psi_{j,k}(P) = \EE \big[ F_j ( t_k \mid A=1, L) - F_j ( t_k \mid
A=0, L)\big]\) for \(j=1,2\) and \(k=1, 2, 3\).

Table \ref{table:cr:1:rel:mse} shows the relative mean squared error
of estimators for each parameters for the \(\sigma_n\)-weighted
one-step estimator compared to the unweighted one-step estimator, the
$\Sigma_n$-weighted one-step estimator, and the iteratively targeted
estimator from \cite{rytgaard2021estimation}. When using the
iteratively targeted estimator, we constructed targeted estimators for
the average treatment effects on both subdistributions (for each
time-point) separately, and computed the effect on the survival
probability by summing over the effects on the subdistributions.

\begin{table}[ht]
\centering
\begin{adjustbox}{tabular=r|ccc|ccc|cccc,center}
  & \multicolumn{3}{c}{\(\hat{F}_1\)} &
  \multicolumn{3}{|c}{\(\hat{F}_2\)}
  &   \multicolumn{3}{|c}{\(\hat{S}\)} \\
  \hline
  \(t_k\)  & 0.6 & 0.8 & 1 & 0.6 & 0.8 & 1 & 0.6 & 0.8 & 1 \\
  \hline\hline
  $\sigma_n$-weighted & 1.000 & 1.000 & 1.000 & 1.000 & 1.000 & 1.000 & 1.000 & 1.000 & 1.000 \\
  $\Sigma_n$-weighted & 1.054 & 1.068 & 1.034 & 1.077 & 1.046 & 1.017 & 1.079 & 1.135 & 1.094 \\
  unweighted & 1.001 & 0.997 & 1.001 & 0.996 & 1.000 & 1.001 & 1.002 & 1.001 & 1.002 \\
  iterative & 1.047 & 1.064 & 1.034 & 1.055 & 1.044 & 1.020 & 1.064 &   1.139 & 1.134
\end{adjustbox}
\caption{Results from a simulation study with sample size \(n=200\)
  and \(M=500\) repetitions. Shown are the relative mean squared
  errors across the \(M=500\) simulation repetitions for the
  \(\sigma_n\)-weighted one-step estimator, the $\Sigma_n$-weighted
  one-step estimator, the unweighted one-step estimator and the
  iteratively targeted estimator from \cite{rytgaard2021estimation}
  all relative to the \(\sigma_n\)-weighted one-step estimator.  }
\label{table:cr:1:rel:mse}
\end{table}

\section{Concluding remarks}
\label{sec:concluding:remarks}

The main contribution of the presented work is the methodology to
construct semiparametric efficient plug-in estimators simultaneously
targeting all target parameters. As far as we are concerned, there is
no other method to achieve this.  Another important result is that we
can get inference for the full survival curve across time by targeting
over a grid which is fine enough; this result is particularly useful
from a practical perspective, allowing for much simpler and faster
implementations.

Our simulation studies demonstrate the potential benefits achieved
when using our one-step compared to the iterative procedure, both in
terms of compatibility (probabilities summing to one, and monotonicity
of survival and competing risks curves), but also potentially in terms
of finite-sample mean squared error gains. The results of Table
\ref{table:surv:1:rel:mse} and Table \ref{table:cr:1:rel:mse} give
some indications of this.

We further found in our simulations different advantages of the
different options for the Hilbert space norm used to construct the
one-step targeted update. Indeed, Table \ref{table:surv:1:rel:mse}
shows a slightly lower mean squared error of the unweighted compared
to the variance-weighted, whereas Table \ref{table:cr:1:rel:mse} shows
a lower mean squared error of the covariance-weighted compared to the
variance-weighted. On the other hand, we generally found the
\(\Sigma_n\)-weighted to converge much faster than both the
\(\sigma_n\)-weighted and the unweighted one-step estimators, and this
option may thus be preferred in many situations, especially with
larger sample sizes where the mean squared errors gain are
diminished. Seemingly, it can differ from one situation to another
what is better. This will be investigated further in future work.

\newpage

\appendix

\section*{Appendix A}
\renewcommand\theequation{A.\arabic{equation}}

\subsection*{Proof of Theorem \ref{thm:targeting:grid}}

\begin{proof} We consider an expansion as follows:
  \begin{align}
    \mathbb{P}_n D^*_t (\hat{P}^*_n)
    &=
      \mathbb{P}_n D^*_t (\hat{P}^*_n)- \mathbb{P}_n D^*_{t_{m(t)}} (\hat{P}^*_n)  + o_P(n^{-1/2})\notag\\
    &= (  \mathbb{P}^*_n - P_0 )  \big( D^*_{t} (\hat{P}^*_n) - D^*_{t_{m(t)}} (\hat{P}^*_n) \big)
      + P_0 \big( D^*_{t} (\hat{P}^*_n) - D^*_{t_{m(t)}} (\hat{P}_n^*) \big) + o_P(n^{-1/2})\notag\\
    \begin{split}
    &= \underbrace{ (  \mathbb{P}_n - P_0 )  \big( D^*_{t} (\hat{P}^*_n) - D^*_{t_{m(t)}} (\hat{P}^*_n) \big)}_{
      (*)
      }
      + P_0 \big( D^*_{t} (\hat{P}^*_n) - D^*_{t_{m(t)}} (\hat{P}^*_n) \big) + o_P(n^{-1/2})\\[-1.4em]
    &\qquad\qquad\qquad\qquad\qquad\qquad\qquad\qquad\qquad\quad\qquad
      - P_0 \big( D^*_{t} (P_0) - D^*_{t_{m(t)}} (P_0) \big);
      \end{split}
      \label{eq:expan1}
  \end{align}
  note that at the first equality we subtracted
  \(\mathbb{P}_n D^*_{t_{m(t)}} (\hat{P}^*_n) \) which is
  \(o_P(n^{-1/2})\) per \eqref{eq:eic:eq:max}, and at the third
  equality we added \(P_0 D^*_{t_{m(t)}} (P_0) \) and subtracted
  \(P_0 D^*_{t} (P_0) \) which are both zero. Now, consider \((*)\)
  above. By \citet[][Lemma 19.24]{van2000asymptotic} and Assumption
  A1, we have that this is \(o_P(n^{-1/2}) \).  Defining
  \(f_{n,0} (t) := P_0 ( D^*_{t} (\hat{P}^*_n) - D^*_{t} (P_0) ) \),
  we continue with the remaining terms of \eqref{eq:expan1}:
  \begin{align*}
    &    P_0 \big( D^*_{t} (\hat{P}^*_n) - D^*_{t_{m(t)}} (\hat{P}_n) \big) 
      - P_0 \big( D^*_{t} (P_0) - D^*_{t_{m(t)}} (P_0)\big) \\
    &\qquad =
      P_0 \big( D^*_{t} (\hat{P}^*_n) -  D^*_{t} (P_0) \big) 
      - P_0 \big(D^*_{t_{m(t)}} (\hat{P}_n) - D^*_{t_{m(t)}} (P_0)\big)  \\
    & \qquad = f_{n,0} (t) - f_{n,0} (t_{m(t)}), 
  \end{align*}
  where evaluation of \(f_{n,0} (t)\) reveals that
    \begin{align*}
      f_{n,0} (t) & =
                    \EE_{P_0} \bigg[  \sum_{l=1}^J \bigg(
                    \int_0^\tau   h_{1,l,t,s} (\hat{P}^*_n) (O) \, \big( N_l(ds) -
                    \1 \lbrace \tilde{T}\ge s\rbrace \hat{\lambda}^*_l(s \, \vert \, A,L) ds \big)   \\[-0.4em]
                  &\qquad\qquad\qquad -
                    \int_0^\tau   h_{1,l,t,s} (P_0) (O) \, \big( N_l(ds) -
                    \1 \lbrace \tilde{T}\ge s\rbrace \lambda_{0,l}(s \, \vert \, A,L) ds \big) \bigg)
                    \bigg] \\
                  &=  \EE_{P_0} \bigg[   \sum_{l=1}^J \int_0^\tau  \big(
                    h_{1,l,t,s} (\hat{P}^*_n) (O) -
                    h_{1,l,t,s} (P_0) (O) \big) \, \big( N_l(ds) -
                    \1 \lbrace \tilde{T}\ge s\rbrace \lambda_{0,l}(s \, \vert \, A,L) ds \big)\\[-0.4em]
                  &\qquad\qquad\qquad\qquad\qquad\quad -   \sum_{l=1}^J
                    \int_0^{\tilde{T}\wedge \tau }   h_{1,l,t,s} (\hat{P}^*_n) (O) \,     \big(
                    \hat{\lambda}^*_l(s \, \vert \, A,L) -
                    \lambda_{0,l}(s \, \vert \, A,L) \big)ds \bigg]. 
                    \intertext{The first line of the right hand side above is a martingale integral of
                    a predictable process; accordingly, this term is simply zero, and we get: }
                  &= -  \EE_{P_0} \bigg[ \sum_{l=1}^J
                    \int_0^{\tilde{T}\wedge \tau }  h_{1,l,t,s} (\hat{P}^*_n) (O) \, \big(
                    \hat{\lambda}^*_l(s \, \vert \, A,L) -
                    \lambda_{0,l}(s \, \vert \, A,L) \big)ds\bigg].
    \end{align*}
    Thus, we have that
    \begin{align*}
      & f_{n,0} (t) -  f_{n,0} (t_{m(t)}) \\
      & \,\, =  \sum_{l=1}^J \int_{\mathcal{L}} \sum_{a=0,1}
        \int_0^\tau   \big( h_{1,l,t_{m(t)},s} - h_{1,l,t,s} \big) (\hat{P}^*_n) (a,\ell) \, \big(
        \hat{\lambda}^*_l(s \, \vert \, a,\ell) -
        \lambda_{0,l}(s \, \vert \, a,\ell) \big)ds \big)  \pi_0 (a \mid \ell) d\mu_0(\ell), 
    \end{align*}
    and the Cauchy-Schwarz inequality immediately yields the following
    bound
\begin{align}
  \big\vert f_{n,0} (t) -  f_{n,0} (t_{m(t)}) \big\vert \le
  \sum_{l=1}^J
  \big\Vert (h_{1,l,t_{m(t)}} - h_{1,l,t} ) (\hat{P}^*_n) \big\Vert_{\mu_0 \otimes \pi_0 \otimes \rho} \,
  \big\Vert \hat{\lambda}^*_l -
  \lambda_{0,l} \big\Vert_{\mu_0 \otimes \pi_0 \otimes \rho},
      \label{eq:cs:bound}
\end{align}
By Assumption A2 we have that the first factor is bounded by
\( \vert t_{m(t)} - t \vert^{1/2}\). Thus we have that
\begin{align*}
  \big\Vert (h_{1,l,t_{m(t)}} - h_{1,l,t} ) (\hat{P}^*_n) \big\Vert_{\mu_0 \otimes \pi_0 \otimes \rho}
  = o_P( n^{-1/6 - \eta/2}), 
\end{align*}
by the assumption that the grid size goes to zero faster than
\(n^{-1/3-\eta}\), for some \(\eta >0\).  Furthermore, when using the
highly adaptive lasso estimator for each cause-specific hazard
\citep{rytgaard2021estimation} we have that
\begin{align*}
  \big\Vert \hat{\lambda}^*_l - \lambda_{0,l} \big\Vert_{\mu_0 \otimes \pi_0 \otimes \rho} = o_P (n^{-1/3-\eta}
  ). 
\end{align*}
We conclude that
\( \mathbb{P}_n D^*_t (\hat{P}^*_n) = o_P( n^{-1/2})\) which completes
the proof.
\end{proof}

\section*{Appendix B}
\renewcommand\theequation{B.\arabic{equation}}

\subsection*{On Assumption A2 of Theorem \ref{thm:targeting:grid}}

We here verify that Assumption A2 of Theorem \ref{thm:targeting:grid}
holds for our considered setting { under the assumption that
\begin{align}
  \sqrt{ \int_{\mathcal{L}} \sum_{a=0,1} \big( F_1 ( t_{m(t)} \mid a, \ell) - F_1 ( t
  \mid a, \ell) \big)^2 \pi_0 ( a, \ell) d\mu_0 (\ell)} \le K' \vert t_{m(t)} - t \vert^{1/2},
  \label{eq:ass:lipschiptz}
\end{align}
for a constant \(K'>0\).} Recall that Assumption A2 states that
\begin{align*}
  \big\Vert ( h_{1,l,t_{m(t)}} - h_{1,l,t} ) (\hat{P}^*_n)
  \big\Vert_{\mu_0 \otimes \pi_0 \otimes \rho} \le \vert t_m -t \vert^{1/2},
\end{align*}
where \( \Vert \cdot \Vert_{\mu_0 \otimes \pi_0 \otimes \rho} \)
denotes the \(L_2( {\mu_0 \otimes \pi_0 \otimes \rho} )\)-norm, i.e.,
\begin{align*}
  \Vert f\Vert_{\mu_0 \otimes \pi_0 \otimes \rho} = \sqrt{\int_{\mathcal{L}} \sum_{a=0,1}
  \int_0^{\tau} \big( f(t, a, \ell)\big)^2 dt \pi_0(a \mid \ell) d \mu_0(\ell) }.
\end{align*}
Let us consider what
\( \big( h_{1,l,t_{m(t)},s} - h_{1,l,t,s} \big) (P) (A,L)\) looks
like:
\begin{align*}
  \big( h_{1,l,t_{m(t)},s} - h_{1,l,t,s} \big)  (P)(A,L) =
  \frac{
  \1 \lbrace A=a^*\rbrace }{ \hat{\pi} (A \, \vert \, L) }   \frac{1}{
  S^c( s- \, \vert \, A,L)} \qquad\qquad\qquad\qquad\qquad\qquad\qquad\qquad\qquad \\
   \begin{cases}
     \1 \lbrace s \le t_{m(t)}\rbrace \big( 1- \frac{ F_1 ( t_{m(t)}
       \mid A, L) - F_1( s\mid A, L)}{ S(s \, \vert \, A,L)}\big) - \1
     \lbrace s \le t\rbrace \big( 1- \frac{ F_1 ( t \mid A, L) - F_1(
       s\mid A, L)}{
       S(s \, \vert \, A,L)}\big), &\text{when }\, l=1, \\
     \1 \lbrace s \le t\rbrace \big( \frac{ F_1 ( t \mid A, L) - F_1 (
       s \mid A, L)}{ S(s \, \vert \, A,L) }\big)- \1 \lbrace s \le
     t_{m(t)}\rbrace \big( \frac{ F_1 ( t_{m(t)}\mid A, L) - F_1 ( s
       \mid A, L)}{ S(s \, \vert \, A,L) }\big) , &\text{when }\, l
     \neq 1.
  \end{cases}
\end{align*}
Consider first the case that \(l=1\),
\begin{align}
  &  \1 \lbrace s \le t_{m(t)}\rbrace \bigg( 1- \frac{ F_1 ( t_{m(t)}
    \mid A, L) - F_1( s\mid A, L)}{ S(s \, \vert \, A,L)}\bigg) - \1
    \lbrace s \le t\rbrace \bigg( 1- \frac{ F_1 ( t \mid A, L) - F_1(
    s\mid A, L)}{
    S(s \, \vert \, A,L)}\bigg) \notag\\
  \begin{split}
  & \qquad = \big(  \1 \lbrace s \le t_{m(t)}\rbrace -  \1
    \lbrace s \le t\rbrace \big)  \bigg( 1 + \frac{  F_1(
    s\mid A, L)}{
    S(s \, \vert \, A,L)} \bigg) \\
  & \qquad\qquad\qquad\qquad\qquad - \, \frac{ 1}{
    S(s \, \vert \, A,L)} \Big(
    \1 \lbrace s \le t_{m(t)}\rbrace F_1 ( t_{m(t)}
    \mid A, L)
     - \1
     \lbrace s \le t\rbrace F_1 ( t \mid A, L) \Big),
     \label{eq:l:1}
   \end{split}
\end{align}
and then, likewise, the case that \(l\neq 1\),
\begin{align}
  &  \1 \lbrace s \le t\rbrace \bigg( \frac{ F_1 ( t \mid A, L) - F_1 (
       s \mid A, L)}{ S(s \, \vert \, A,L) }\bigg)- \1 \lbrace s \le
     t_{m(t)}\rbrace \bigg( \frac{ F_1 ( t_{m(t)}\mid A, L) - F_1 ( s
       \mid A, L)}{ S(s \, \vert \, A,L) }\bigg)  \notag\\
  \begin{split}
    & \qquad = \big( \1 \lbrace s \le t_{m(t)}\rbrace - \1 \lbrace s
    \le t\rbrace \big) \bigg( \frac{ F_1 ( s \mid A, L)}{ S(s \, \vert
      \, A,L) }\bigg)\\
  & \qquad\qquad\qquad\qquad\qquad - \, \frac{ 1}{
    S(s \, \vert \, A,L)} \Big(
    \1 \lbrace s \le t_{m(t)}\rbrace F_1 ( t_{m(t)}
    \mid A, L)
     - \1
     \lbrace s \le t\rbrace F_1 ( t \mid A, L) \Big).
     \label{eq:l:2}
   \end{split}
\end{align}
Assume with no loss of generality that \(t_{m(t)}\ge t\), and see that
\begin{align*}
\1 \lbrace s \le t_{m(t)}\rbrace - \1 \lbrace s \le t\rbrace = \1
  \lbrace s \in (t, t_{m(t)}]\rbrace.
\end{align*}
Since \(S(t\mid A, L) > \kappa' >0\) (Assumption \ref{ass:M:cadlag}),
both expressions in \eqref{eq:l:1} and \eqref{eq:l:2} are really
driven by the term
\begin{align}
  &    \1 \lbrace s \le t_{m(t)}\rbrace F_1 ( t_{m(t)} \mid A, L) -
    \1 \lbrace s \le t\rbrace F_1 ( t \mid A, L)  \notag\\
   & \qquad = \big ( \1 \lbrace s \le t_{m(t)}\rbrace - \1 \lbrace s \le t\rbrace\big)
    F_1 ( t_{m(t)} \mid A, L)
     +  \1 \lbrace s \le t\rbrace \big( F_1 ( t_{m(t)} \mid A, L) - F_1 ( t \mid A, L)  \big)\notag \\
  &      \qquad = \1 \lbrace s \in (t, t_{m(t)}]\rbrace F_1 ( t_{m(t)} \mid A, L)
    +  \1 \lbrace s \le t\rbrace \big( F_1 ( t_{m(t)} \mid A, L) - F_1 ( t \mid A, L)  \big) .
        \label{eq:dom:term}
\end{align}
Collecting what we have above now yields that
\begin{align*}
  &  \big( h_{1,l,t_{m(t)},s} - h_{1,l,t,s} \big)  (P)(A,L) 
  \\
  & \qquad  \le    \frac{
    \1 \lbrace A=a^*\rbrace }{ \hat{\pi} (A \, \vert \, L) }   \frac{1}{
    S^c( s- \, \vert \, A,L)} \bigg( \1
    \lbrace s \in (t, t_{m(t)}]\rbrace
    \big( 1+ \kappa'^{-1}\big)  \\
  &\qquad\qquad - \, \kappa'^{-1}\Big(
    \1 \lbrace s \in (t, t_{m(t)}]\rbrace F_1 ( t_{m(t)} \mid A, L)
    +  \1 \lbrace s \le t\rbrace \big( F_1 ( t_{m(t)} \mid A, L) - F_1 ( t \mid A, L)  \big)\Big)\bigg), 
\end{align*}
so that,
\begin{align*}
  &\big\Vert ( h_{1,l,t_{m(t)}} - h_{1,l,t} ) (\hat{P}^*_n)
    \big\Vert_{\mu_0 \otimes \pi_0 \otimes \rho} \\
  & \quad
    \le \eta \sqrt{ \int_{\mathcal{L}} \sum_{a=0,1} \bigg( \int_0^{\tau}
    \big(h_{1,l,t_{m(t)},s} - h_{1,l,t,s} \big)^2 (\hat{P}^*_n) (o) ds
    \bigg) \pi_0( a\mid \ell) d\mu_0 (\ell) } \\
  & \quad \le \eta \big( 1+ \kappa'^{-1}\big) \sqrt{ \int_0^{\tau}
    \1
    \lbrace s \in (t, t_{m(t)}]\rbrace
    ds} \\[-0.3cm]
  &\qquad\qquad\qquad\qquad + \, \eta\kappa'^{-1}\sqrt{  \int_{\mathcal{L}} \sum_{a=0,1} \bigg( \int_0^{\tau}
    \1 \lbrace s \in (t, t_{m(t)}]\rbrace
    \big( F_1 ( t_{m(t)} \mid a, \ell)\big)^2
    ds \bigg) \pi_0(a\mid \ell) d\mu_0(\ell)}\\
  &\quad\qquad\qquad\,\, + \, \eta\kappa'^{-1}\sqrt{  \int_{\mathcal{L}} \sum_{a=0,1} \bigg( \int_0^{\tau}
    \1 \lbrace s \le t\rbrace \big( F_1 ( t_{m(t)} \mid a, \ell) - F_1 ( t \mid a, \ell)  \big)^2  ds \bigg) \pi_0(a\mid \ell) d\mu_0(\ell)}\\
  &\quad = \eta\big( 1+ \kappa'^{-1}\big)
    \big( t_{m(t)} - t \big)^{1/2} +
    \,\eta\kappa'^{-1} \sqrt{  \big( t_{m(t)} - t \big)\int_{\mathcal{L}} \sum_{a=0,1}  \big( F_1 ( t_{m(t)} \mid a, \ell)\big)^2
    \pi_0(a\mid \ell) d\mu_0(\ell)} \\
  & \qquad\qquad\quad\qquad\qquad\qquad\quad +
    \, \eta\kappa'^{-1}  t^{1/2} \sqrt{  \int_{\mathcal{L}} \sum_{a=0,1} 
    \big( F_1 ( t_{m(t)} \mid a, \ell) - F_1 ( t \mid a, \ell)  \big)^2 \pi_0(a\mid \ell) d\mu_0(\ell)}\\
  & \quad \le \eta \big( 1+ \kappa'^{-1}\big) \big( t_{m(t)} - t \big)^{1/2}
    + \eta \kappa'^{-1}\big( t_{m(t)} - t \big)^{1/2}
    + \eta \kappa'^{-1} t^{1/2}  K' \big( t_{m(t)} - t \big)^{1/2}
  \\
  & \quad \le \eta \big( 1+ ( 2 + K') \kappa'^{-1} \big) \big( t_{m(t)} - t \big)^{1/2}
    ,    
\end{align*}
by application of the assumption stated in Equation
\eqref{eq:ass:lipschiptz}. This establishes the claim.

\subsection*{On Assumption A1 of Theorem \ref{thm:targeting:grid}}

Under condition (ii) from Section
\ref{sec:conditions:inference:infinite} we can replace Assumption A1
stating that
\begin{align*}
  P_0 \big( D^*_{t}(\hat{P}^*_n) - D^*_{t_{m(t)}}(\hat{P}^*_n)
  \big)^2\overset{P}{\rightarrow} 0, 
\end{align*}
by:
\begin{align}
&  \sup_{t} \,\, P_0 \big(
  D^*_{t}(\hat{P}^*_n) - D^*_{t}(P_0) \big)^2 \overset{P}{\rightarrow} 0 ,
                \label{eq:alternive:assumption:L2:1} \\
  & \sup_{\mathclap{t}}  \, \,\, \, P_0 \big(
    D^*_{t_{m(t)}}(P_0) - D^*_{t}(P_0) \big)^2 \overset{P}{\rightarrow} 0.
  \label{eq:alternive:assumption:L2:2}
\end{align}
This follows since
\begin{align*}
  & P_0 \big(  D^*_{t}(\hat{P}^*_n) - D^*_{t_{m(t)}}(\hat{P}^*_n) \big)^2\\
  & \quad =  P_0 \big(
    D^*_{t}(\hat{P}^*_n) - D^*_{t}(P_0) +
    D^*_{t_{m(t)}} (P_0) - D^*_{t_{m(t)}}(\hat{P}^*_n) +   D^*_{t}(P_0) - D
    ^*_{t_{m(t)}}(P_0)\big)^2 \\
  & \quad \le    P_0 \big(
    D^*_{t}(\hat{P}^*_n) - D^*_{t}(P_0) \big)^2 + P_0 \big( 
    D^*_{t_{m(t)}} (P_0) - D^*_{t_{m(t)}}(\hat{P}^*_n) \big)^2
    + P_0 \big( 
    D^*_{t} (P_0) - D^*_{t_{m(t)}}(P_0) \big)^2 ,
\end{align*}
where the second term is \(o_P(n^{-1/2})\) under condition (ii) from
Section \ref{sec:conditions:inference:infinite} and the remaining
terms are \(o_P(n^{-1/2})\) by
\eqref{eq:alternive:assumption:L2:1}--\eqref{eq:alternive:assumption:L2:2}.

\section*{Appendix C}
\renewcommand\theequation{C.\arabic{equation}}

\subsection*{Verifying the bound \eqref{eq:max:bounded:by:norm} for
  the covariance-weighted norm}

We demonstrate the bound \eqref{eq:max:bounded:by:norm} from Section
\ref{sec:practical:implementation:infinite} of the main text holds for
the covariance-weighted Hilbert space norm proposed in Remark
\ref{rem:choice:Sigma:covariance:norm}. For this purpose, assume that
there is at least one \(t_m\) such that
\( \vert \mathbb{P}_n D_{t_m}^*(P_{\eps}) \vert / \sigma_{t_m} > c_n
\Vert \mathbb{P}_n D^* (P_\eps) \Vert_{\Sigma_{M_n}}\) where
\(c_n = \sqrt{ M_n}\). For the positive definite \(M_n\times M_n\)
matrix \(\Sigma_{M_n}\) we have that
  \begin{align}
    M_n  (\mathbb{P}_n D^* (P_\eps))^\top  \Sigma_{M_n}^{-1} \mathbb{P}_n D^* (P_\eps)
    \ge     (\mathbb{P}_n D^* (P_\eps))^\top  D_{M_n}^{-1} \mathbb{P}_n D^* (P_\eps),
    \label{eq:linalg}
  \end{align}
  where \(D_{M_n}\) is the diagonal matrix with diagonal equal to the
  diagonal of \(\Sigma_{M_n}\).  Now \eqref{eq:linalg} directly
  implies that
  \begin{align*}
    \Vert \mathbb{P}_n D^* (P_\eps) \Vert_{\Sigma_{M_n}}
    \ge  \frac{1}{\sqrt{ M_n}}  \sqrt{ \sum_{m=1}^{M_n} \big\vert   \mathbb{P}_n
    D_{t_m}^*(P_{\eps}) \big\vert^2 / \sigma_{t_m}^2 } > \Vert \mathbb{P}_n D^* (P_\eps) \Vert_{\Sigma_{M_n}}, 
  \end{align*}
  contradicting the assumption that
  \( \vert \mathbb{P}_n D_{t_m}^*(P_{\eps}) \vert / \sigma_{t_m} > c_n
  \Vert \mathbb{P}_n D^* (P_\eps) \Vert_{\Sigma_{M_n}}\) with
  \(c_n = \sqrt{ M_n}\) for some \(t_m\). Thus,
  \eqref{eq:max:bounded:by:norm} from Section
  \ref{sec:practical:implementation:infinite} of the main text holds
  for the covariance-weighted Hilbert space norm with
  \(c_n=\sqrt{ M_n}\).

  Verifying \eqref{eq:linalg} comes down to showing that
  \(M_n \Sigma_{M_n}^{-1} - D_{M_n}^{-1}\) is non-negative
  definite. The matrix
  \(\tilde{\Sigma}_{M_n} = D_{M_n}^{-1/2} \Sigma_{M_n}
  D_{M_n}^{-1/2}\) is positive definite with all diagonals elements
  equal to 1, and we can write
  \(M_n \Sigma_{M_n}^{-1} - D_{M_n}^{-1} = D_{M_n}^{-1/2} \big( M_n
  \tilde{\Sigma}_{M_n}^{-1} - \mathbb{I}_{M_n} \big) D_{M_n}^{-1/2} \)
  where \(\mathbb{I}_{M_n}\) is the \(M_n\times M_n\) identity
  matrix. The sum of the eigenvalues of \(\tilde{\Sigma}_{M_n}\)
  equals the trace of \(\tilde{\Sigma}_{M_n}\), which is
  \(M_n\). Thus, all eigenvalues of \(\tilde{\Sigma}_{M_n}\) belongs
  to \((0,n)\) and all eigenvalues of \(\tilde{\Sigma}_{M_n}^{-1}\)
  belongs to \((1/n,\infty)\). In conclusion,
  \(M_n \tilde{\Sigma}_{M_n}^{-1} - \mathbb{I}_{M_n} \) is positive
  definite.

  \section*{Appendix D}
\renewcommand\theequation{D.\arabic{equation}}

\subsection*{Additional options for implementation}

Recall from Section \ref{sec:implementation} that we proceed
recursively with small update steps for each \(l\) defining as follows
\begin{align}
  \hat{\lambda}_{l, dx} &= \hat{\lambda}_{l} \exp\bigg(  \frac{
                            \big( \mathbb{P}_n \tilde{D}^*(\hat{\lambda}_{n}, \hat{\pi}_n, \hat{S}_n^c) \big)^\top \Sigma_{d}^{-1}
                            \tilde{h}_{l, t} (\hat{\lambda}_{n}, \hat{\pi}_n, \hat{S}_n^c)
                            }{
                            \Vert \mathbb{P}_n
                            \tilde{D}^*(\hat{\lambda}_{n}, \hat{\pi}_n, \hat{S}_n^c) \Vert_{\Sigma_d} 
                            }\bigg) ,\label{eq:recursive:steps:1}
                            \intertext{and, for \(m\ge 1\), }
                            \hat{\lambda}_{l, (m+1)dx}
  &= \hat{\lambda}_{l,mdx}  \exp\bigg( \frac{
    \big( \mathbb{P}_n \tilde{D}^*(\hat{\lambda}_{mdx}, \hat{\pi}_n, \hat{S}_n^c)  \big)^\top
    \Sigma^{-1} \tilde{h}_{l, t} (\hat{\lambda}_{mdx}, \hat{\pi}_n, \hat{S}_n^c)
    }{
    \Vert \mathbb{P}_n
    \tilde{D}^*(\hat{\lambda}_{mdx}, \hat{\pi}_n, \hat{S}_n^c) \Vert_{\Sigma_d} 
    }\bigg) ;  \label{eq:recursive:steps:2}
\end{align}
these recursive update steps are continued until for a given \(m^*\)
we have
\begin{align*}
  \max_{j,k} \,
  \frac{ \mathbb{P}_n \tilde{D}_{j,k}^*(\hat{\lambda}_{m^* dx}, \hat{\pi}_n, \hat{S}^c_n
  )}{\hat{\sigma}_{j,k}} \le \frac{1}{\sqrt{n} \log n } . 
\end{align*}
However, we also have the option to update along
\eqref{eq:recursive:steps:1}--\eqref{eq:recursive:steps:2}
\textit{only} in the directions that have not currently been solved.
Recall that
\begin{align*} 
  \mathbb{P}_n \tilde{D}^*(\hat{\lambda}_{m dx}, \hat{\pi}_n, \hat{S}_n^c)
  =  \big(\mathbb{P}_n \tilde{D}^*_{j,t_k} (\hat{\lambda}_{m dx}, \hat{\pi}_n, \hat{S}_n^c) \, : \, j=1, \ldots, J, k=1, \ldots, K\big).
\end{align*}
For a fixed \(m\), denote by \(\mathcal{J}_m, \mathcal{K}_m\) the set
of indices such that
\begin{align*}
  \frac{ \mathbb{P}_n \tilde{D}_{j,k}^*(\hat{\lambda}_{m dx}, \hat{\pi}_n, \hat{S}^c_n
  )}{\hat{\sigma}_{j,k}} \le \frac{1}{\sqrt{n} \log n }, \quad
  \text{for } \,\, j\in \mathcal{J}_m, k \in \mathcal{K}_m.  
\end{align*}
Let \( 0_{\mathcal{J}_m,\mathcal{K}_m}\) denote the zero vector of
length corresponding to the number of elements in \(\mathcal{J}_m\)
and \(\mathcal{K}_m\). Now, when we update along
\eqref{eq:recursive:steps:1}--\eqref{eq:recursive:steps:2}, we do not
include contributions from \(j\in \mathcal{J}_m, k\in \mathcal{K}_m\);
effectively, this means that we substitute the vector
\begin{align*} 
  \big((\mathbb{P}_n \tilde{D}^*_{j,t_k} (\hat{\lambda}_{m dx},
  \hat{\pi}_n, \hat{S}_n^c), 0_{\mathcal{J}_m,\mathcal{K}_m})
  \, : \, j\not\in\mathcal{J}_m, k\not\in\mathcal{K}_m \big).
\end{align*}
for
\(\mathbb{P}_n \tilde{D}^*(\hat{\lambda}_{m dx}, \hat{\pi}_n,
\hat{S}_n^c)\) in
\eqref{eq:recursive:steps:1}--\eqref{eq:recursive:steps:2}. At the
next round, we likewise define the set of indices
\(\mathcal{J}_{m+1}, \mathcal{K}_{m+1}\) such that
\begin{align*}
  \frac{ \mathbb{P}_n \tilde{D}_{j,k}^*(\hat{\lambda}_{(m+1) dx}, \hat{\pi}_n, \hat{S}^c_n
  )}{\hat{\sigma}_{j,k}} \le \frac{1}{\sqrt{n} \log n }, \quad
  \text{for } \,\, j\in \mathcal{J}_{m+1}, k \in \mathcal{K}_{m+1}, 
\end{align*}
and use
\begin{align*} 
  \big((\mathbb{P}_n \tilde{D}^*_{j,t_k} (\hat{\lambda}_{{m+1} dx},
  \hat{\pi}_n, \hat{S}_n^c), 0_{\mathcal{J}_{m+1},\mathcal{K}_{m+1}})
  \, : \, j\not\in\mathcal{J}_{m+1}, k\not\in\mathcal{K}_{m+1}  \big),
\end{align*}
for updating. Note that there is no guarantee that
\(\mathcal{J}_{m+1} \subseteq \mathcal{J}_{m}\) nor
\( \mathcal{K}_{m+1}\subseteq \mathcal{K}_{m}\), i.e., a particular
direction may be solved for a given step but then not solved for later
steps. However, we only proceed with a given update step if the norm
\( \Vert \mathbb{P}_n \tilde{D}^*(\hat{\lambda}_{mdx}, \hat{\pi}_n,
\hat{S}_n^c) \Vert_{\Sigma_d} \) is not increasing from \(m\) to
\(m+1\); otherwise we decrease the step size and proceed from
here. With this approach we repeat the recursive update steps until
\(\mathcal{J}_{m^*} = \lbrace 1, \ldots, J\rbrace\) and
\( \mathcal{K}_{m^*}= \lbrace 1, \ldots, K\rbrace\).
  
\newpage

\end{document}